\documentclass[a4paper,12pt]{article}

\title{Linear complementary dual \\quasi-cyclic codes of index 2}
\author{Kanat Abdukhalikov \\	
Department of Mathematical Sciences, \\
UAE University, PO Box 15551, Al Ain, UAE\\
Email: abdukhalik@uaeu.ac.ae \bigskip  \\  
Duy Ho \\
Department of Mathematics and Statistics, \\
UiT The Arctic University of Norway, Tromsø
9037, Norway\\
Email: duyho92@gmail.com \bigskip \\
San Ling \\	
School of Physical and Mathematical Sciences, \\
Nanyang Technological University, Singapore 637371\\ 
 Email: lingsan@ntu.edu.sg\bigskip  \\ 
Gyanendra K. Verma \\	
Department of Mathematical Sciences, \\
UAE University, PO Box 15551, Al Ain, UAE\\
Email:  gkvermaiitdmaths@gmail.com}
\date{ }

\usepackage{amsthm,amsmath,amssymb} 
\DeclareMathOperator{\rank}{rank}
\usepackage{url} 
\usepackage{cases} 
\usepackage{makecell}
\usepackage{enumerate}
\usepackage{tikz-cd}
\usepackage{float}

\begin{document} 

\maketitle

\theoremstyle{plain} 
\newtheorem{lemma}{Lemma}[section] 
\newtheorem{theorem}[lemma]{Theorem}
\newtheorem{corollary}[lemma]{Corollary}
\newtheorem{proposition}[lemma]{Proposition}

\theoremstyle{definition}
\newtheorem{definition}{Definition}[section] 
\newtheorem{remark}{Remark}
\newtheorem{example}{Example}

\newcommand{\eps}{\varepsilon}
\newcommand{\inprod}[1]{\left\langle #1 \right\rangle}
\newcommand{\la}{\lambda} 
\newcommand{\al}{\alpha}
\newcommand{\om}{\omega} 
\newcommand{\gam}{\gamma}
\newcommand{\be}{\beta}
\newcommand{\sig}{\sigma}

\begin{abstract} 
We provide a polynomial approach to investigate linear complementary dual (LCD) quasi-cyclic codes over finite fields. We establish necessary and sufficient conditions for  LCD quasi-cyclic codes of index 2 with respect to the Euclidean, Hermitian, and symplectic inner products.  As a consequence of these characterizations, we derive necessary and sufficient conditions for LCD one-generator quasi-cyclic codes. Furthermore, using these characterizations, we construct some new quasi-cyclic LCD codes over small fields.
\end{abstract}



\section{Introduction}

The family of quasi-cyclic codes over finite fields is an important class of linear codes that generalizes cyclic codes. The study of quasi-cyclic codes can be traced back to the late 1960s, beginning with the paper by  Townsend and Weldon \cite{townsend1967}, and the works by Karlin \cite{karlin1969, karlin1970}.
In those early days, quasi-cyclic codes were already known to be asymptotically good, see for example \cite{chen1969}.
Many constructions of quasi-cyclic codes contain codes with optimal parameters, as shown in \cite{grassltable} and \cite{gulliver1991}. 

 In the  2000s, Ling and Sol\'e studied  the algebraic structure of quasi-cyclic codes in a series of articles \cite{ling2001,ling2003,ling2005,ling2006}.    In \cite{lally2001}, Lally and Fitzpatrick proved  that every quasi-cyclic code has a generating set of polynomials in the form of a reduced Gr\"obner basis. Based on these structural properties,  more  asymptotic results, minimum distance bounds, and further applications of quasi-cyclic codes were obtained in the literature.  To name a few, we refer to the paper Semenov and Trifonov \cite{semenov2012} on the spectral method for  quasi-cyclic codes, see also by other authors in \cite{luo2024} and \cite{zeh2016}.  
Applications of quasi-cyclic codes in constructing quantum codes have become a very active research topic in recent years, see for example  \cite{ezerman2025}, \cite{galindo2018}, and \cite{guan2024}.  

 Linear codes with complementary duals (LCD codes) were introduced by Massey in \cite{massey}.
In \cite{sendrier1997},  Sendrier proved that LCD codes are asymptotically good and used
them in relation to equivalence testing of linear codes in \cite{sendrier2000}.
 Recently, LCD codes became an attractive research interest as they offer solutions
 to many cryptographic problems, for example against side-channel attacks and fault non-invasive attacks, see \cite{carlet2016}. 
 Several constructions of LCD codes with optimal parameters are known, see for example \cite{abdukhalikov2025, bouyuklieva2021, carlet2018mds, liu2021, sok2018, wu2021}.
 In \cite{carlet2018}, it was shown that  any linear code over $\mathbb{F}_q$ ($q  > 3$) is equivalent to a Euclidean LCD code and any linear code over $\mathbb{F}_{q^2}$ ($q  > 2$) is equivalent to a Hermitian LCD code. 
Analogous results were considered for symplectic self-dual  codes in \cite{li2024b} and vector rank metric codes in \cite{ho2025}.

  In 1994, a characterization  for LCD cyclic codes  in terms of their generator polynomials was provided by Yang and Massey in \cite{yang1994}.  
 For the case of quasi-cyclic codes,  Esmaeili and Yari \cite{esmaeili2009} provided   a sufficient condition
 for quasi-cyclic codes to be Euclidean LCD codes and gave a method for constructing quasi-cyclic Euclidean LCD codes.
 In 2016, G\"uneri, \"Ozkaya  and Sol\'e in  \cite{guneri2016}  characterized Euclidean LCD quasi-cyclic codes using the Chinese Remainder Theorem (CRT) decomposition of codes introduced by Ling and Sol\'e in \cite{ling2001}. 

One-generator quasi-cyclic codes were studied in \cite{quasi1},  \cite{seguin2004} and \cite{seguin1990}.  The general structure of quasi-cyclic codes of index 2 was examined in \cite{quasi2}.
  Recently in \cite{guneri2023} and \cite{guan2023}, characterizations of  LCD one-generator quasi-cyclic codes of index $\ell$  were obtained.

 Based on the previous results of  \cite{quasi2}, \cite{guneri2016}, \cite{lally2001}, and \cite{ling2001}, in this paper we provide a new characterization  for Euclidean, Hermitian and symplectic LCD quasi-cyclic codes of index 2 in terms of  generating sets of polynomials. 
 Our results extend
the existing characterizations obtained in \cite{guneri2023} and \cite{guan2023} for one-generator quasi-cyclic codes of index 2.

 The content of the paper is organized as follows.  In Section 2, we recall preliminary results from linear codes and quasi-cyclic codes. In Section 3, we present a new characterization for Euclidean LCD quasi-cyclic codes of index 2. In Section 4, we consider the special case of Euclidean LCD one-generator quasi-cyclic codes of index 2.   
In Sections 5 and 6 these results were generalized for symplectic and Hermitian   quasi-cyclic codes of index 2. 

 All computations in this paper have been done with the computer algebra system MAGMA \cite{magma}.

\section{Preliminaries}


\subsection{Background on linear and quasi-cyclic codes}
 Let $F=\mathbb{F}_q$ denote the finite field with $q$ elements, where $q$ is a prime power.
 A \textit{linear code} $C$ of length $n$ is a subspace of the vector  space $F^n$.  The elements of $C$ are codewords. The \textit{Euclidean hull} of $C$ is defined as
 \[
\text{Hull}(C):= C \cap C^{\perp_e},
 \]
where $C^{\perp_e}$ denotes the dual of $C$ with respect to the usual Euclidean inner product. We remind the reader that the Euclidean inner product of $\mathbf{x},\mathbf{y} \in F^n$ with  $\mathbf{x} = (x_1, \dots, x_n), \mathbf{y} = (y_1, \dots, y_n)$ is given as
\[
\langle \mathbf{x}, \mathbf{y} \rangle_e = \sum_{i=1}^n x_i y_i.
\]
If $C \cap C^{\perp_e} =\{ \mathbf{0}\}$, then we say that $C$ is a \textit{linear code with complementary dual}. Here, the dual is defined using the Euclidean inner product, and we will abbreviate such a code as \textit{Euclidean LCD}.


 Let $T$ be the standard cyclic shift operator on $F^n$. 
 A linear code is said to be \textit{quasi-cyclic of index $\ell$ (QC)} if it is invariant under $T^\ell$. 
 We assume that $\ell$ divides $n$. If $\ell= 1$, then the QC code is a cyclic code.

Let $R = F[x]/\langle x^m-1\rangle$. We recall that cyclic codes of length $m$ over $F$ can be considered as ideals of $R$.

 Let $n = m\ell$ and let $C$ be a linear quasi-cyclic code of length $m \ell$ and index $\ell$ over $F$. Let
 \[
\mathbf{c} =(c_{0,0},  c_{0,1}, \dots,  c_{0,\ell-1},c_{1,0},  c_{1,1}, \dots  c_{1,\ell-1}, \dots, c_{m-1,0},  c_{m-1,1}, \dots,  c_{m-1,\ell-1})
 \]
denote a codeword in $C$. Define a map $\varphi: F^{m\ell} \rightarrow R^\ell$ by
\[
\varphi(\mathbf{c}) = (c_0(x), c_1(x), \dots, c_{\ell-1}(x)) \in R^\ell,
\]
where
\[
c_j(x) = c_{0,j}+c_{1,j}x+c_{2,j}x^2+ \dots + c_{m-1,j}x^{m-1} \in R.
\]
The following lemma is well-known.
 \begin{lemma}[\cite{lally2001,ling2001}] The map  $\varphi$ induces a one-to-one correspondence between quasi-cyclic
 codes over $F$ of index $\ell$ and length $m \ell$ and linear codes over $R$ of length $\ell$.
\end{lemma}

\subsection{Decomposition of quasi-cyclic codes}\label{subsection2.2}

 
Let $f(x) = a_0+a_1x+a_2x+ \dots + a_k x^k$ be a polynomial of degree $k$. The reciprocal polynomial of $f(x)$ is the polynomial
\[
f^*(x)=x^{\deg{f(x)}} f(x^{-1}) = a_k+a_{k-1}x+a_{k-2}x+ \dots + a_0 x^k.
\] 
A polynomial $f(x)$ is said to be self-reciprocal if $f(x)$ and $f^*(x)$ are associates (i.e., $f^*(x)=\alpha f(x)$ for some $\alpha\in F$). 
Let $f(x) = a_0+a_1x+a_2x+ \dots + a_{m} x^{m} \in \mathbb{F}_q[x]$, where $m$ is as before, that is $R= F[x]/\langle x^m-1 \rangle$. 
The   \textit{transpose polynomial} of $f(x)\in \mathbb{F}_q[x]$ is the polynomial
\[
\bar{f} (x) =x^{m} f(x^{-1}) = a_m+a_{m-1}x+a_{m-2}x+ \dots + a_0 x^m.
\]
Then $\bar{f}(x)=x^{m-\deg f(x)}f^*(x)$. Assume that $\gcd(q,m)=1$. With this assumption, we have the following factorization into distinct irreducible polynomials in $\mathbb{F}_q[x]$:

$$
x^m-1= \delta \prod_{i=1}^sf_i(x) \prod_{j=1}^p h_j(x)h^*_j(x),
$$
where $\delta$ is nonzero in $\mathbb{F}_q$, $f_i(x)$ is self-reciprocal for all $1\le i \le s$,  $h_j(x)$   and $h_j^*(x)$ are reciprocal pairs for all $1 \le j \le p$.

For each $i$ and $j$, let $F_i=F[x]/  (f_i )$, 
$H_j'= F[x]/ (h_j )$,   and
$H_j''= F[x]/ (h_j^* )$. 
Let $\xi$ be a primitive $m^{\text{th}}$ root of unity over $F$.
Let $\xi^{u_i}$ and  $\xi^{v_j}$ be roots of $f_i(x)$ and $h_j(x)$, respectively.
Then we also have $h_j^*(\xi^{-v_j})=0$, and
$F_i \cong F(\xi^{u_i})$, $H_j' \cong F(\xi^{v_j})$, and $H_j'' \cong F(\xi^{-v_j})$, see \cite[p. 136]{guneri2021}.

The map $\bar\ : f(x) \mapsto \bar{f}(x)$ can be naturally extended to the following isomorphisms:   
\begin{equation}
	\begin{split}
		\bar\  :\  \mathbb{F}_{q}[x]/ (f_i(x) )&\to \mathbb{F}_{q}[x]/ (f_i (x)) ,\\
		\bar\  :\  \mathbb{F}_{q}[x]/ (h_j(x) )&\to \mathbb{F}_{q}[x]/ (h_j^*(x) ) .\\
	\end{split}
\end{equation} 
Therefore, the map $\bar\ $ is an isomorphism from $H_j'=\mathbb{F}_q[x]/ ( h_j(x))$ to 
$H_j''=\mathbb{F}_q[x]/ ( h_j^*(x) )$.

By the Chinese Remainder Theorem (CRT), $R$ can be decomposed as
 $$
R \cong \left(\bigoplus_{i=1}^{s} F_i \right) 
\oplus
\left(\bigoplus_{j=1}^{p} \left( H_j'
\oplus  H_j'' \right) \right).
 $$ 
The isomorphism between $R$ and its CRT decomposition is given by 
\[
a(x) \mapsto \left(a(\xi^{u_1}), \dots, a(\xi^{u_s}),a(\xi^{v_1}),a(\xi^{-v_1}), \dots, a(\xi^{v_p}),a(\xi^{-v_p})\right). 
\]

As $\xi$ is an $m^{\mathrm{th}}$ root of unity, we have $\xi^m=1$. Thus  $a(\xi^{-1})=\bar{a}(\xi)$ for all polynomials $a(x)$ of degree at most $m$. Hence the above isomorphism can be written as 
\[
a(x) \mapsto \left(a(\xi^{u_1}), \dots, a(\xi^{u_s}),a(\xi^{v_1}),\bar{a}(\xi^{v_1}), \dots, a(\xi^{v_p}),\bar{a}(\xi^{v_p})\right)). 
\]
This isomorphism extends naturally to  $R^\ell$, which implies that
 $$
R^\ell \cong \left(\bigoplus_{i=1}^{s} F_i^\ell \right) 
\oplus
\left(\bigoplus_{j=1}^{p} \left( (H_j')^\ell
\oplus  (H_j'')^\ell  \right) \right).
 $$ 
Then, a QC code $C$ of index $\ell$ can be  decomposed as 
\begin{equation} \label{crtC}
C \cong \left(\bigoplus_{i=1}^{s} C_i \right) 
\oplus
\left(\bigoplus_{j=1}^{p} \left(  C_j' 
\oplus   C_j''   \right) \right), 
\end{equation} 
where each component code is a linear code of length $\ell$ over the base field $(F_i, H_j'$ or $H_j'')$ it is defined.  The component codes  $C_i, C_j', C_j''$ are called the \textit{constituents} of $C$.

The constituents can be described in terms of the generators of $C$. Namely, if $ C$ is an $r$-generator QC code with generators
\[
 \{ (a_{1,1}(x),\dots ,a_{1,\ell}(x)) ,\dots, (a_{r,1}(x),\dots,a_{r,\ell} (x) ) \} \subset R^\ell,
\]
then 
\begin{align*}
C_i &= \text{Span}_{F_i} \{ (a_{k,1}(\xi^{u_i}),\dots ,a_{k,\ell}(\xi^{u_i})) : 1 \le k \le r \}, \text{ for } 1 \le i \le s, \\
C_j' &= \text{Span}_{H_j'} \{ (a_{k,1}(\xi^{v_j}),\dots ,a_{k,\ell}(\xi^{v_j})) : 1 \le k \le r \}, \text{ for } 1 \le j \le p, \\
C_j'' &= \text{Span}_{H_j''} \{ (\bar{a}_{k,1}(\xi^{v_j}),\dots ,\bar{a}_{k,\ell}(\xi^{v_j})) : 1 \le k \le r \}, \text{ for } 1 \le j \le p.
\end{align*}

With a QC code $C$ and its CRT decomposition given in \eqref{crtC}, the Euclidean dual of $C$ is of the form 
\begin{equation} \label{crtCdual}
C^{\perp_e} = \left(\bigoplus_{i=1}^{s} C_i^{\perp_h} \right) 
\oplus
\left(\bigoplus_{j=1}^{p} \left(  C_j''^{\perp_e}
\oplus   C_j'^{\perp_e}   \right) \right).
\end{equation} 
Here, ${\perp_h}$ denotes the Hermitian dual on $F_i^\ell$ (for each $1\le i \le s$) with respect to  the Hermitian inner product
\begin{equation}\label{hinnerproduct}
	\langle \textbf{c},\textbf{d} \rangle_h= \sum_{k=1}^\ell c_k(\xi^{u_i})\bar{d}_k(\xi^{u_i}),
\end{equation}
where
$
\textbf{c} = (c_1(\xi^{u_i}), \dots,  c_\ell(\xi^{u_i})), \textbf{d} = (d_1(\xi^{u_i}), \dots,  d_\ell(\xi^{u_i})) \in F_i^\ell.
$ 
This is the inner product induced by $x\mapsto x^{-1}$, not the usual Hermitian inner product, see \cite[p. 2693]{ling2005}.  For each $1 \le j \le p$, the vector space $(H_j')^\ell \cong (H_j'')^\ell$ is equipped with the usual Euclidean inner product and ${\perp_e}$ denotes the usual Euclidean dual.

\begin{remark} \label{remarkH} 
Since $f_i(x)$ is self-reciprocal, the cardinality of $F_i$, say $q_i$, is an even power of $q$ for all $1 \le i \le s$ with two exceptions. 
One of these exceptions, for all $m$ and $q$, is the field coming from the irreducible factor $x-1$ of $x^m-1.$ 
When $q$ is odd and $m$ is even, $x+1$ is another self-reciprocal irreducible factor of $x^m-1$. In these cases, $q_i=q$. Except for two cases,
the Hermitian inner product $\langle \cdot,\cdot \rangle_h$ is equivalent to the usual Hermitian inner product, see also \cite[p.72]{guneri2016}  and \cite[p.136]{guneri2021}.  For the two exceptions, in which case the corresponding field $F_i$ is $F$, we equip $F_i^\ell$ with the  usual Euclidean inner product. Then the previous formula for $\langle \cdot,\cdot \rangle_h$ is still true, since $\xi^{u_i}=\xi^{-u_i}=\pm1$.


\end{remark}

We have the following characterization of Euclidean LCD QC codes from \cite[Theorem 7.3.6]{guneri2021}.

\begin{theorem} \label{qclcd} Let $C$ be a $q$-ary QC code of length $m\ell$ and index $\ell$ with a CRT decomposition as in \eqref{crtC}.
Then $C$ is Euclidean LCD if and only if 
$C_i \cap C_i^{\perp_h} = \{\mathbf{0}\}$ for all $1 \le i \le s$, and $C_j' \cap C_j''^{\perp_e} = \{\mathbf{0}\}$, $C_j'' \cap C_j'^{\perp_e} = \{\mathbf{0}\}$ 
for all $1\le j \le p$. 
 \end{theorem}

\subsection{Quasi-cyclic codes of index 2}

In \cite{lally2001}, Lally and Fitzpatrick showed that a quasi-cyclic code of index $\ell$ can be generated by the rows of an upper triangular $\ell \times \ell$  polynomial matrix satisfying certain conditions. For the case $\ell=2$, this result was improved in    \cite[Theorem 3.1]{quasi2} to the  following theorem.
\begin{theorem} \label{qc2}
Let $C$ be a quasi-cyclic code of length $2m$ and index $2.$ 
Then $C$ is generated by two elements $(g_{11}(x), g_{12}(x))$ and $(0, g_{22}(x))$ such that they satisfy the following conditions: 
\begin{gather}  
g_{11}(x) \mid (x^m-1) \text{ and } g_{22}(x) \mid(x^m-1), \notag \\
\deg g_{12}(x) < \deg g_{22}(x), \tag{$\ast$}\\
g_{11}(x)g_{22}(x) \mid (x^m-1)g_{12}(x). \notag
\end{gather}
 Moreover, in this case $\dim C = 2m-\deg g_{11}(x)-\deg g_{22}(x)$.
\end{theorem}

\begin{remark} \label{remarkgcd}   If $\gcd(q,m)=1$, then the condition 
$$g_{11}(x)g_{22}(x) \mid (x^m-1)g_{12}(x)$$ 
in Theorem \ref{qc2} is equivalent to the condition $\gcd(g_{11}(x),g_{22}(x)) \mid g_{12}(x)$, since $x^m-1$ has no multiple roots, see  \cite[Remark 3.1]{quasi2}.
\end{remark}


\begin{lemma} \label{lemma1gen} 
Let $\gcd(q,m)=1$ and let $C$ be a quasi-cyclic code generated by one element $(g_{11}(x),g_{12}(x))$, where $g_{11}(x) \mid (x^m-1)$.  
Let $g(x)=\gcd(g_{11}(x),g_{12}(x))$, $g_{11}(x)=g(x)g_{11}'(x)$, $g_{12}(x)=g(x)g_{12}'(x)$. 
Let 
\[
g_{22}(x) = \dfrac{x^m-1}{g_{11}'(x)}.
\]
Then the following statements are true.

1. The code $C$ is generated by two elements $(g_{11}(x),g_{12}(x) \mod{g_{22}(x)})$ and $(0,g_{22}(x))$ satisfying Conditions  $(\ast)$.

2. $\gcd(g_{11}(x),g_{22}(x))=g(x)$.
\end{lemma}
\begin{proof}  Let $C$ be   generated by one element $(g_{11}(x),g_{12}(x))$. Then
\[
\dfrac{x^m-1}{g_{11}(x)} (g_{11}(x),g_{12}(x)) =  \left(0, \dfrac{x^m-1}{g_{11}'(x)}  g_{12}'(x)\right).
\]
The zeros of the cyclic code $\left\langle  \dfrac{x^m-1}{g_{11}'(x)}  g_{12}'(x) \right\rangle$ are the same as the zeros of the polynomial $\dfrac{x^m-1}{g_{11}'(x)}$, so 
\[
\left\langle   \dfrac{x^m-1}{g_{11}'(x)}  g_{12}'(x) \right\rangle = \left\langle  \dfrac{x^m-1}{g_{11}'(x)}\right\rangle. 
\]
Thus $C$ is generated by the elements $(g_{11}(x),g_{12}(x))$ and $(0,g_{22}(x))$. Moreover, 
\[
\gcd(g_{11}(x),g_{22}(x))=g(x).
\]
Finally,  we can reduce $g_{12}(x)$  modulo $g_{22}(x)$ to the reduced Gr\"obner basis form, see \cite{lally2001}.
\end{proof}

\begin{remark} In \cite[Lemma 1]{seguin2004}, S\'eguin showed that if  $C$ is a quasi-cyclic code generated by one element $(g_{11}(x),g_{12}(x))$ with $g_{11}(x) \mid x^m-1$ and $g(x)=\gcd(g_{11}(x),g_{12}(x))$, then  $\dim C = m-\deg g(x)$. With the choice of $g_{22}(x)$ described in Lemma \ref{lemma1gen}, we see that the dimension of $C$ in Theorem \ref{qclcd} is consistent with the result by S\'eguin. 
\end{remark}

\section{Euclidean LCD quasi-cyclic codes of index 2}
From now on, we will assume that $\gcd(q,m)=1$. Let $C$ be a quasi-cyclic code   of index $2$. Then by Theorem \ref{qc2}, $C$ is generated by two elements $(g_{11}(x), g_{12}(x))$ and $(0, g_{22}(x))$ satisfying Conditions ($\ast$).  Since $\gcd(q,m)=1$, the code $C$ can be decomposed using the Chinese Remainder Theorem (CRT) as described in Subsection 2.2. In this setting, each constituent of $C$ is generated by the rows of a $2 \times 2$ matrix over its field of definition. Explicitly, $C_i, C_j'$ and $C_j''$ are generated by the rows of the matrices
$$
G_i=  \begin{bmatrix}
g_{11}(\xi^{u_i}) & g_{12}(\xi^{u_i}) \\
0 & g_{22}(\xi^{u_i}) 
\end{bmatrix} ,
G_j'=   \begin{bmatrix}
g_{11}(\xi^{v_j}) & g_{12}(\xi^{v_j}) \\
0 & g_{22}(\xi^{v_j}) 
\end{bmatrix}  ,
G_j''=   \begin{bmatrix}
\bar{g}_{11}(\xi^{v_j}) & \bar{g}_{12}(\xi^{v_j}) \\
0 & \bar{g}_{22}(\xi^{v_j}) 
\end{bmatrix}  ,
$$
respectively.

Let $g(x)= \gcd(g_{11}(x),g_{22}(x))$. Since we are assuming $\gcd(q,m)=1$, the condition 
$$g_{11}(x)g_{22}(x) \mid (x^m-1)g_{12}(x)$$ 
in Theorem \ref{qc2} is equivalent to the condition $g(x) \mid g_{12}(x)$, see Remark \ref{remarkgcd}.

Let  $l(x)= (x^m-1)/\text{lcm}(g_{11}(x),g_{22}(x))$. Let $g_{11}(x)=g(x)   g_{11}'(x), g_{22}(x)=g(x) g_{22}'(x)$, and
\begin{align*}
g_{11}'(x)=r_{11}(x) t_{11}(x),\\
g_{22}'(x)=r_{22}(x) t_{22}(x),
\end{align*}
where $r_{11}(x)=\gcd(g'_{11}(x), g'^*_{11}(x))$, and $r_{22}(x)=\gcd(g'_{22}(x), g'^*_{22}(x))$. Then $r_{11}(x)$ and $r_{22}(x)$ are self-reciprocal. 
The following is the main theorem of this section.

\begin{theorem} \label{main} Let $C$ be a quasi-cyclic code   of index $2$. Let $(g_{11}(x), g_{12}(x))$ and $(0, g_{22}(x))$ be the generators of $C$ satisfying Conditions $(\ast)$.   
	Then $C$ is Euclidean LCD if and only if all of the following conditions are true:
	\begin{enumerate}[(I)]
		\item  $g$ is self-reciprocal.
		\item  $l$ is self-reciprocal. 
		\item  $\gcd(t_{22}(x),g_{12}(x))=1$. 
		\item $\gcd(r_{22}(x), g_{11}(x)\bar{g}_{11}(x)+g_{12}(x) \bar{g}_{12}(x))=1$.
	\end{enumerate}
\end{theorem}

The proof of Theorem \ref{main} is presented at the end of the section following Lemmata \ref{nI}, \ref{nII}, \ref{nIII}, \ref{nIV}, \ref{s1}, and \ref{s2}.

\begin{lemma} \label{nI} If $C$ is Euclidean LCD, then  (I) holds. 
\end{lemma}
\begin{proof} Assume that  $g$ is not self-reciprocal. This implies that there exists $h_j$ such that $h_j\mid g$ and  $h_j^*  \nmid g$. 

1. Since $h_j \mid g$ and $g\mid g_{12}$, it follows that $h_j\mid g_{12}$. Then $g_{11}(\xi^{v_j})= g_{12}(\xi^{v_j})=g_{22}(\xi^{v_j}) =0$, and $C_j'$ is generated by the rows of the  matrix
\[
G_j'=\begin{bmatrix}
g_{11}(\xi^{v_j}) & g_{12}(\xi^{v_j}) \\
0 & g_{22}(\xi^{v_j}) 
\end{bmatrix}    = \begin{bmatrix}
0 & 0 \\
0 & 0 
\end{bmatrix}. 
\]
This implies that $C_j'= \{ \mathbf{0} \}$, and so ${C_j'}^{\perp_e}$ is $2$-dimensional over  the base field $H_j''= F[x]/\langle h^*_j \rangle$. Therefore, ${C_j'}^{\perp_e}=(H_j'')^2.$

2. On the other hand,  the condition $h_j^* \nmid g$ implies that $h_j^* $ does not divide at least one of $g_{11}$ and $g_{22}$.  
This means that $C_j''$ is at least 1-dimensional over $H_j''= F[x]/\langle h_j^* \rangle$.  It follows that $C_j'' \cap {C_j'}^{\perp_e}=C_j'' \ne \{\mathbf{0}\}$. By Theorem \ref{qclcd}, $C$ is not Euclidean LCD. This proves the lemma.
\end{proof}

\begin{lemma} \label{nII}   If $C$ is Euclidean LCD, then  (II) holds. 
\end{lemma}
\begin{proof} Assume that  $l$ is not self-reciprocal. This implies that there exists $h_j$ such that $h_j\mid l$ and  $h_j^*  \nmid l$.  Since $h_j \mid l$, it follows that $h_j \nmid \text{lcm}(g_{11},g_{22})$. In particular, $h_j\nmid g_{11}$ and $h_j\nmid g_{22}$. Then    $g_{11}(\xi^{v_j}) \ne 0$, $g_{22}(\xi^{v_j}) \ne 0$ and $\rank(G_j')=2$.
This implies that  $C_j'=(H_j')^2.$ 

On the other hand, $h_j^* \nmid l$ implies that $h_j^* \mid \text{lcm}(g_{11},g_{22}) = g \cdot g_{11}' \cdot g_{22}'$. But since $h_j \nmid g$ (from the condition $h_j \mid l$ above), by Lemma \ref{nI} we also have that $h_j^* \nmid g$. Hence $h_j^* \mid g_{11}'\cdot g_{22}'$, which means $h_j^*$ divides  either  $g_{11}$  or $g_{22}$ but not both. This implies that $C_j''$ is 1-dimensional. Then $C_j''^{\perp_e}$ is also 1-dimensional, and so $C_j' \cap C_j''^{\perp_e} \ne \{ \mathbf{0} \}$.  By Theorem \ref{qclcd}, $C$ is not Euclidean LCD.
\end{proof}

\begin{lemma} \label{nIII} Assume that (I) and (II) hold. If (III) is not true, then $C$ is not Euclidean LCD. 
\end{lemma}

\begin{proof} We note that for all $1 \le i \le s$, the irreducible polynomial $f_i$ is self-reciprocal, and so $f_i \nmid t_{22}$.  
Assume that (III) is not true, that is $\gcd(t_{22},g_{12}) \ne 1$. Then there exists $h_j$ such that $h_j \mid \gcd(t_{22},g_{12})$.

1. Since  $g$ and $l$ are self-reciprocal, 
\[
x^m-1  = g \cdot l \cdot r_{11} \cdot t_{11} \cdot r_{22} \cdot t_{22}.  
\] 
Furthermore, the  reciprocal of  $x^m-1$ is  $-(x^m-1)$, so we can rewrite 
\begin{align*}
x^m-1   &= \alpha \cdot g \cdot l \cdot r_{11} \cdot t_{11}^*    \cdot  r_{22} \cdot  t_{22}^*,
\end{align*} 
for some $\alpha 
\in F$. In particular, since  $h_j \mid t_{22}$, and since the polynomials $g,l,r_{11},r_{22}$ are self-reciprocal, we have that $h_j^* \mid t_{11}$.

2. Since $h_j \mid t_{22}$, it follows that $h_j \mid g_{22}$ and $h_j \nmid g_{11}$.  Then  
\[
G_j'=\begin{bmatrix}
g_{11}(\xi^{v_j}) & g_{12}(\xi^{v_j}) \\
0 & g_{22}(\xi^{v_j}) 
\end{bmatrix} = \begin{bmatrix}
g_{11}(\xi^{v_j}) & 0 \\
0 & 0 
\end{bmatrix},
\]
where $g_{11}(\xi^{v_j}) \ne 0$.
Since $h_j^* \mid t_{11}$, we have that $h_j^* \mid g_{11}$ and $h_j^* \nmid g_{22}$. Then 
\[
G_j'' = \begin{bmatrix}
	\bar{g}_{11}(\xi^{v_j}) & \bar{g}_{12}(\xi^{v_j}) \\
	0 & \bar{g}_{22}(\xi^{v_j}) 
\end{bmatrix} = \begin{bmatrix}
 	0 &  \bar{g}_{12}(\xi^{v_j}) \\
	0 & \bar{g}_{22}(\xi^{v_j}) 
\end{bmatrix},
\]
where $\bar{g}_{22}(\xi^{v_j}) \ne 0$. Hence  $C_j' =\langle (1,0) \rangle$, $C_j'' =\langle (0,1) \rangle$, and $C_j''^{\perp_e} =\langle (1,0) \rangle = C_j'$.
Therefore, $C_j' \cap {C_j''}^{\perp_e}=C_j' \ne \{\mathbf{0}\}$ and so $C$ is not Euclidean LCD.
\end{proof}

\begin{lemma}   \label{nIV}  Assume that (I) and (II) hold. If (IV) is not true, then $C$ is not Euclidean LCD. 
\end{lemma}

\begin{proof}
Similar to the proof of Lemma \ref{nIII}, if $g$ and $l$ are self-reciprocal, we have
\[
x^m-1  = g \cdot l \cdot r_{11} \cdot t_{11} \cdot r_{22} \cdot t_{22}.
\] 
Assume that condition (IV) is not true, that is, there exists an irreducible factor  $a(x)$ of $x^m-1$ such that 
$$a \mid \gcd(r_{22}, g_{11}\bar{g}_{11}+g_{12} \bar{g}_{12}).$$
Since $a \mid r_{22}$, we have that  $a \nmid g_{11}$,  $a \mid g_{22}$,  $a \nmid g_{11}^*$,  and $a \mid g_{22}^*$.   We have two cases depending on whether $a$ is self-reciprocal. 

1. $a$ is self-reciprocal, that is, $a=f_i$ for some $i$. Since $f_i  \nmid g_{11}$ and $f_i \mid g_{22}$, $G_i$ is of the form
\[
G_i=  \begin{bmatrix}
g_{11}(\xi^{u_i}) & g_{12}(\xi^{u_i}) \\
0 & 0
\end{bmatrix},
\] 
where $g_{11}(\xi^{u_i}) \ne 0$. Then   $C_i =\langle (g_{11}(\xi^{u_i}), g_{12}(\xi^{u_i}))\rangle$ is a 1-generator code, whose dual is $C_i^{\perp_h} =\langle (-\bar{g}_{12}(\xi^{u_i}),\bar{g}_{11}(\xi^{u_i}))\rangle$, see Remark \ref{remarkH}.  But since $f_i \mid  g_{11}\bar{g}_{11}+g_{12} \bar{g}_{12}$, 
\[
 g_{11}(\xi^{u_i})\bar{g}_{11}(\xi^{u_i})+g_{12}(\xi^{u_i}) \bar{g}_{12}(\xi^{u_i}) =0, 
\]
and so $C_i = C_i^{\perp_h}$. Hence $C_i \cap  C_i^{\perp_h} \ne \{\mathbf{0}\}$, and so $C$ is not Euclidean LCD. 

2. $a$ is not self-reciprocal, that is, $a=h_j$ for some $j$. With the same reasoning as in case 1, we have that $C_j' =\langle (g_{11}(\xi^{v_j}), g_{12}(\xi^{v_j}))\rangle$ and $C_j'^{\perp_e} =\langle (-g_{12}(\xi^{v_j}), g_{11}(\xi^{v_j}))\rangle$. Also, $C_j'' =\langle (\bar{g}_{11}(\xi^{v_j}), \bar{g}_{12}(\xi^{v_j}))\rangle$. 
But since $h_j \mid g_{11}\bar{g}_{11}+g_{12} \bar{g}_{12}$, 
\[
 g_{11}(\xi^{v_j})\bar{g}_{11}(\xi^{v_j})+g_{12}(\xi^{v_j}) \bar{g}_{12}(\xi^{v_j}) =0, 
\]
and so $C_j'' = C_j'^{\perp_e}$. Hence $C_j'' \cap  C_j'^{\perp_e} \ne \{\mathbf{0}\}$, and so $C$ is not Euclidean LCD. \qedhere

\end{proof}

\begin{lemma} \label{s1}  If (I), (II), (III) and (IV) hold, then $C_i \cap C_i^{\perp_h} = \{\mathbf{0}\}$ for all $1 \le i \le s$. 
\end{lemma}
\begin{proof} For each $i$, there are four cases depending on the divisibility of $f_i$ with respect to $g_{11}$ and $g_{22}$. We will show that $C_i \cap C_i^{\perp_h} = \{\mathbf{0}\}$ in each of these four cases. 
	
1. $f_i \mid g_{11}$ and $f_i \mid g_{22}$. This implies that $f_i \mid g$ and hence $f_i \mid g_{12}$. Then $G_i$ is the zero matrix, $C_i= \{\mathbf{0}\}$, and so $C_i \cap C_i^{\perp_h} = \{\mathbf{0}\}$. 
	
2. $f_i \mid g_{11}$ and $f_i \nmid g_{22}$. Then $g_{11}(\xi^{u_i})=0$, $g_{22}(\xi^{u_i})\ne0$,  and $G_i$ is of the form
\[
G_i= \begin{bmatrix}
	0 & g_{12}(\xi^{u_i}) \\
	0 & g_{22}(\xi^{u_i}) 
\end{bmatrix}.
\] 
Then $C_i=\langle (0,1) \rangle$, $C_i^{\perp_h}=\langle (1,0) \rangle$ and so $C_i \cap C_i^{\perp_h} = \{\mathbf{0}\}$.

3. $f_i \nmid g_{11}$ and $f_i \mid g_{22}$.  Then $g_{11}(\xi^{u_i})\ne 0$, $g_{22}(\xi^{u_i})= 0$, and $G_i$ is of the form
\[
G_i=  \begin{bmatrix}
g_{11}(\xi^{u_i}) & g_{12}(\xi^{u_i}) \\
0 & 0
\end{bmatrix} .
\] 
Hence   $C_i =\langle (g_{11}(\xi^{u_i}), g_{12}(\xi^{u_i}))\rangle$ is a 1-generator code, whose dual is 
\[
C_i^{\perp_h} =\langle (-\bar{g}_{12}(\xi^{u_i}),\bar{g}_{11}(\xi^{u_i}))\rangle.
\]
From condition (IV), we have 
$
\gcd(r_{22}, g_{11}\bar{g}_{11}+g_{12}\bar{g}_{12})=1,
$
and since $f_i \mid r_{22}$, it follows that $f_i \nmid (g_{11}\bar{g}_{11}+g_{12}\bar{g}_{12})$. Then
$$
g_{11}(\xi^{u_i})\bar{g}_{11}(\xi^{u_i})+g_{12}(\xi^{u_i}) \bar{g}_{12}(\xi^{u_i})) \ne 0,
$$
 and $C_i \cap C_i^{\perp_h} = \{\mathbf{0}\}$.

4. $f_i \nmid g_{11}$ and $f_i \nmid g_{22}$. In this case, $C_i$ is $2$-dimensional, $C_i^{\perp_h} = \{\mathbf{0}\}$, and so $C_i \cap C_i^{\perp_h} = \{\mathbf{0}\}$.  \end{proof}

\begin{lemma} \label{s2}  If (I), (II), (III) and (IV) hold, then $C_j' \cap C_j''^{\perp_e} = \{\mathbf{0}\}$ and $C_j'' \cap C_j'^{\perp_e} = \{\mathbf{0}\}$  for all $1\le j \le p$.  
\end{lemma}

\begin{proof} Since  $g$ and $l$ are self-reciprocal, we have that
\begin{align*}
x^m-1   &= g \cdot l \cdot r_{11} \cdot t_{11}     \cdot  r_{22} \cdot  t_{22}.
\end{align*}
For each $j$, the irreducible polynomial $h_j$ divides exactly one of these six factors of $x^m-1$, which leads to the following six cases.

1. $h_j \mid g$. Since $g \mid g_{12}$, we also have that $h_j \mid g_{12}$.  Then  
$g_{11}(\xi^{v_j})=g_{12}(\xi^{v_j})=g_{22}(\xi^{v_j})=0$. Since $g$ is self-reciprocal, it follows that $h_j^* \mid g$. Similarly as before,
$\bar{g}_{11}(\xi^{v_j})=\bar{g}_{12}(\xi^{v_j})=\bar{g}_{22}(\xi^{v_j})=0$.
Hence $G_j'=G_j''=\{\mathbf{0}\}$, which implies that  $C_j' \cap C_j''^{\perp_e} = \{\mathbf{0}\}$ and $C_j'' \cap C_j'^{\perp_e} = \{\mathbf{0}\}$.

2. $h_j \mid l$. This condition implies that $h_j \nmid g$ and $h_j^* \nmid g$. Then $C_j'$ and $C_j''$ are 2-dimensional, which implies that $C_j'^{\perp_e}$ and $C_j''^{\perp_e}$ are $0$-dimensional. Then $C_j' \cap C_j''^{\perp_e} = \{\mathbf{0}\}$ and $C_j'' \cap C_j'^{\perp_e} = \{\mathbf{0}\}$.

3. $h_j \mid r_{11}$. This condition implies that $h_j \mid g_{11}$, $h_j \nmid g_{22}$, $h_j^* \mid g_{11}$ and $h_j^* \nmid g_{22}$. It follows that  
$$
G_j'=   \begin{bmatrix}
0 & g_{12}(\xi^{v_j}) \\
0 & g_{22}(\xi^{v_j}) 
\end{bmatrix} ,
G_j''=   \begin{bmatrix}
0 & \bar{g}_{12}(\xi^{v_j}) \\
0 & \bar{g}_{22}(\xi^{v_j}) 
\end{bmatrix}.
$$
It can then be readily checked that $C_j' \cap C_j''^{\perp_e} = \{\mathbf{0}\}$ and $C_j'' \cap C_j'^{\perp_e} = \{\mathbf{0}\}$.

4. $h_j \mid t_{11}$. We recall from the proof of Lemma \ref{nIII} that we can  rewrite 
\begin{align*}
x^m-1   &= \alpha \cdot g \cdot l \cdot r_{11} \cdot t_{11}^*    \cdot  r_{22} \cdot  t_{22}^*.
\end{align*}  
Then the condition $h_j \mid t_{11}$   implies that $h_j^* \mid t_{22}$. Hence $h_j \mid g_{11}$, $h_j \nmid g_{22}$, $h_j^* \nmid g_{11}$ and $h_j^* \mid g_{22}$.  
From condition (III), we have that $\gcd(t_{22},g_{12})=1$, and so $h_j^* \nmid g_{12}$.  
The matrices $G_j'$ and  $G_j''$ are of the following form
$$
G_j'=   \begin{bmatrix}
0 & g_{12}(\xi^{v_j}) \\
0 & g_{22}(\xi^{v_j}) 
\end{bmatrix} ,
G_j''=   \begin{bmatrix}
\bar{g}_{11}(\xi^{v_j}) & \bar{g}_{12}(\xi^{v_j}) \\
0 & 0
\end{bmatrix},
$$
where $g_{22}(\xi^{v_j})\ne0$, $\bar{g}_{11}(\xi^{v_j})\ne 0$, and $\bar{g}_{12}(\xi^{v_j})\ne0$.  Then $C_j' =\langle (0,1) \rangle$, $C_j'' =\langle (\bar{g}_{11}(\xi^{v_j}),\bar{g}_{12}(\xi^{v_j})) \rangle$, and the dual codes are 
$$
C_j'^{\perp_e}  =\langle (1,0) \rangle,C_j''^{\perp_e}  =\langle (-\bar{g}_{12}(\xi^{v_j}),\bar{g}_{11}(\xi^{v_j}))\rangle.
$$ 
It can then be readily checked that $C_j' \cap C_j''^{\perp_e} = \{\mathbf{0}\}$ and $C_j'' \cap C_j'^{\perp_e} = \{\mathbf{0}\}$.

5. $h_j \mid r_{22}$. This condition implies that $h_j \nmid g_{11}$, $h_j \mid g_{22}$, $h_j^* \nmid g_{11}$ and $h_j^* \mid g_{22}$. The matrices $G_j'$ and  $G_j''$ are of the following form
$$
G_j'=   \begin{bmatrix}
g_{11}(\xi^{v_j}) & g_{12}(\xi^{v_j}) \\
0 & 0
\end{bmatrix},
G_j''=   \begin{bmatrix}
\bar{g}_{11}(\xi^{v_j}) & \bar{g}_{12}(\xi^{v_j}) \\
0 & 0
\end{bmatrix},
$$
where $g_{11}(\xi^{v_j})\ne0$, and $\bar{g}_{11}(\xi^{v_j})\ne 0$. Then 
$$C_j' =\langle (g_{11}(\xi^{v_j}), g_{12}(\xi^{v_j})) \rangle , C_j'' =\langle (\bar{g}_{11}(\xi^{v_j}), \bar{g}_{12}(\xi^{v_j})) \rangle,$$
and the dual codes are 
$$
C_j'^{\perp_e}  =\langle (-g_{12}(\xi^{v_j}),g_{11}(\xi^{v_j}))\rangle, C_j''^{\perp_e}  =\langle (-\bar{g}_{12}(\xi^{v_j}),\bar{g}_{11}(\xi^{v_j}))\rangle.
$$  
From condition (IV), we have 
$$\gcd(r_{22}, g_{11}(x)\bar{g}_{11}(x)+g_{12}(x) \bar{g}_{12}(x))=1,$$ 
and since $h_j \mid r_{22}$, it follows that $h_j \nmid g_{11}(x)\bar{g}_{11}(x)+g_{12}(x) \bar{g}_{12}(x)$. Then
$$
g_{11}(\xi^{v_j})\bar{g}_{11}(\xi^{v_j})+g_{12}(\xi^{v_j}) \bar{g}_{12}(\xi^{v_j})) \ne 0.
$$
Hence $C_j' \cap C_j''^{\perp_e} = \{\mathbf{0}\}$ and $C_j'' \cap C_j'^{\perp_e} = \{\mathbf{0}\}$.

6. $h_j \mid t_{22}$. This case is similar to case 4.  \qedhere 

\end{proof}

\begin{proof}[Proof of Theorem \ref{main}] We recall that we want to prove that $C$ is Euclidean LCD if and only if the conditions (I), (II), (III) and (IV) hold. The ``if" direction follows from Lemmata \ref{s1} and \ref{s2}. The ``only if" direction follows from Lemmata \ref{nI}, \ref{nII}, \ref{nIII}, and \ref{nIV}.
\end{proof}

We now describe some examples we obtained from Theorem \ref{main}. Following \cite{wang2024}, we denote
\[
d(n, k)_{ \mathbb{F}_q } = \max\{ d(C) : C \text{ is an } [n, k, d(C)] \text{ linear code over } \mathbb{F}_q \},
\]
and
\[
d_{\mathrm{LCD}}(n, k)_{ \mathbb{F}_q } = \max\{ d(C) : C \text{ is an } [n, k, d(C)] \text{ LCD code over } \mathbb{F}_q \}.
\]
We call a linear code with parameters $[n, k, d(n, k)_{\mathbb{F}_q}]$ an \emph{optimal code}. A comprehensive record of $d(n, k)_{\mathbb{F}_q}$   can be found in the online database codetables.de, see ~\cite{grassltable}. An LCD code with parameters $[n, k, d_{\mathrm{LCD}}(n, k)_{\mathbb{F}_q}]$ is called an \emph{optimal LCD code}.

In \cite{carlet2018}, Carlet et al. showed that any linear code is equivalent to a Euclidean LCD code for $q > 3$. This motivates us to search for binary and ternary Euclidean LCD codes with good parameters. 
Using Theorem~\ref{main}, we construct several optimal binary Euclidean LCD codes, as presented in Tables~\ref{tab1} and~\ref{tab2}.

We use the abbreviation \textbf{BKLC} to indicate that a code is the \emph{best known linear code} according to \cite{grassltable}, although it may not be optimal. Similarly, \textbf{BKLCD} denotes the \emph{best known LCD code} according to \cite{li2024} and \cite{wang2024}, while it also may not be optimal.

The  codes obtained in Table \ref{tab1} are compared with the recently updated tables of bounds on the minimum distance of binary Euclidean LCD codes in \cite{li2024} and \cite{wang2024}. In particular, we refer to \cite[Tables~4~and~5]{li2024} for codes of length $29 \le n \le 40$, and to \cite[Table~6]{wang2024} for length $41 \le n \le 50$.

To the best of our knowledge, there is no existing database of LCD codes with length $n > 50$ in the literature. Consequently, the codes obtained in Table~\ref{tab2} are compared with the linear code database available at  codetables.de, see \cite{grassltable}. 
We note that the BKLC codes from \cite{grassltable} used for comparison are not LCD codes.

\begin{table}[H]
\centering
\tiny
\begin{tabular}{|c|c|c|c|c|}
\hline 
$C$ & $g_{11}(x)$ & $g_{12}(x)$ & $g_{22}(x)$ & Remark \\ 
\hline
$[30,15,7]_2$ & $x^2 + x + 1$ & \makecell{$x^{12} + x^{10} + x^9+x$} & 
\makecell{$(x + 1)(x^4 + x + 1)$\\$(x^4 + x^3 + 1)(x^4 + x^3 + x^2 + x + 1)$} & Optimal LCD, see \cite{li2024} \\
\hline
$[30,16,6]_2$ & $x + 1$ & \makecell{ $(x + 1)$\\$(x^{10} + x^5 + x^3 + x^2)$} & 
\makecell{$(x + 1)(x^4 + x + 1)$\\$(x^4 + x^3 + 1)(x^4 + x^3 + x^2 + x + 1)$} & Optimal LCD, see \cite{li2024} \\
\hline
$[34,17,8]_2$ & $x + 1$ & \makecell{$x^{15} + x^{14} + x^{13} + x^9 + x^6$\\$+\,x^5 + x^4 + x^3 + x^2 + 1$} & 
\makecell{$(x^8 + x^5 + x^4 + x^3 + 1)$\\$(x^8 + x^7 + x^6 + x^4 + x^2 + x + 1)$} & Optimal LCD, see \cite{li2024} \\
\hline
$[34,25,4]_2$ & $x + 1$ & $x^7 + x^6 + x^5 + x^4$ & $x^8 + x^5 + x^4 + x^3 + 1$ & Optimal LCD, see \cite{li2024} \\
\hline
$[34,26,4]_2$ & $1$ & $x^6 + x^5 + x^3 + x^2$ & $x^8 + x^5 + x^4 + x^3 + 1$ & Optimal LCD, see \cite{li2024} \\
\hline
$[42,32,4]_2$ & $x^2 + x + 1$ & \makecell{$(x+1)(x^2+x+1)$\\$(x^4 + x^3 + 1)$} & 
\makecell{$(x^2 + x + 1)$\\$(x^3 + x + 1)(x^3 + x^2 + 1)$} & Optimal LCD, see \cite{wang2024} \\
\hline
$[46,23,10]_2$ & $x + 1$ & \makecell{$1+x^2+x^6+x^8$\\$+x^9+x^{13}+x^{15}+x^{20}$} & 
\makecell{$(x^{11} + x^{9} + x^{7} + x^{6} + x^{5} + x + 1)$\\$(x^{11} + x^{10} + x^{6} + x^{5} + x^{4} + x^{2} + 1)$} & BKLCD, improves \cite{wang2024} \\
\hline
$[50,24,10]_2$ & \makecell{$(x+1)(x^4+x^3$\\$+x^2+x+1)$} & \makecell{$(x+1)(1+x^5+$\\$x^7+x^{14})$} & 
\makecell{$(x+1)(x^{20}+x^{15}$\\$+x^{10}+x^5+1)$} & $d=d_{\text{BKLCD}}$, see \cite{wang2024} \\
\hline
$[50,25,10]_2$ & $x + 1$ & \makecell{$x^{22} + x^{21} + x^{20} + x^{18} + x^{14}$\\$+\,x^{13} + x^{12} + x^4 + x^3 + x$} & 
\makecell{$(x^4 + x^3 + x^2 + x + 1)$\\$(x^{20} + x^{15} + x^{10} + x^5 + 1)$} & $d=d_{\text{BKLCD}}$, see \cite{wang2024} \\
\hline
$[50,28,8]_2$ & $x+1$ & \makecell{$(x+1)(1+x^5$\\$+x^7+x^8)$} & 
\makecell{$(x+1)(x^{20}+x^{15}$\\$+x^{10}+x^5+1)$} & $d=d_{\text{BKLCD}}$, see \cite{wang2024} \\
\hline
\end{tabular}
\caption{Binary quasi-cyclic Euclidean LCD codes from Theorem \ref{main}, length $30 \le n \le 50$.}
\label{tab1}
\end{table}

\begin{table}[H]   
    \centering
\tiny\begin{tabular}{|c|c|c|c|c|}
   \hline 
   $C$  & $g_{11}(x)$ & $g_{12}(x)$ & $g_{22}(x)$ & Comparison with \cite{grassltable}  \\ 
   \hline 
   $[54,22,12]_2$ & \makecell{$(x+1)$\\$(x^6+x^3+1)$}  & 
  \makecell{
$x(x + 1)^3(x^6 + x^3 + 1)$\\
$(x^{10} + x^9 + x^6 + x^3 + x^2 + x + 1)$
}& 
   \makecell{$(x+1)$\\$(x^6 + x^3 + 1)$\\$(x^{18} + x^9 + 1)$} & 
   $d = d_{\text{BKLC}} - 1$ \\

   \hline 
   $[54,25,11]_2$ & \makecell{$(x+1)$\\$(x^2+x+1)$} & 
   \makecell{$(x^2+x+1)$\\$(x^{13} + x^9 + x^7 + x^6 + x^5 + 1)$} & 
   \makecell{$(x^2 + x + 1)$\\$(x^6 + x^3 + 1)$\\$(x^{18} + x^9 + 1)$} & 
   $d = d_{\text{BKLC}} - 1$ \\

   \hline 
   $[54,26,11]_2$ & $x^2 + x + 1$ & 
\makecell{
$x^3(x + 1)(x^2 + x + 1)$\\
$(x^3 + x + 1)(x^6 + x^5 + x^4 + x + 1)$\\
$(x^7 + x^6 + x^3 + x + 1)$
}& 
   \makecell{$(x^2 + x + 1)$\\$(x^6 + x^3 + 1)$\\$(x^{18} + x^9 + 1)$} & 
   $d = d_{\text{BKLC}} - 1$ \\

   \hline 
   $[54,27,11]_2$ & $x^2 + x + 1$ & 
   \makecell{$x^{22} + x^{21} + x^{15} + x^{11} + x^{10}$\\$+\,x^8 + x^7 + x^3 + x^2 + 1$} & 
   \makecell{$(x + 1)$\\$(x^6 + x^3 + 1)$\\$(x^{18} + x^9 + 1)$} & 
   $d = d_{\text{BKLC}}$ \\

   \hline 
   $[54,28,9]_2$ & $1$ & 
   \makecell{$x^{23} + x^{19} + x^{12} + x^{10} + x^9$\\$+\,x^8 + x^7 + x^5 + x^2 + x$} & 
   \makecell{$(x^2 + x + 1)$\\$(x^6 + x^3 + 1)$\\$(x^{18} + x^9 + 1)$} & 
   $d = d_{\text{BKLC}} - 1$ \\

   \hline 
   $[54,29,9]_2$ & $1$ & 
   \makecell{$x^{23} + x^{21} + x^{19} + x^{18} + x^{16}$\\$+\,x^{10} + x^6 + x^4 + x^3 + x^2$} & 
   \makecell{$(x + 1)$\\$(x^6 + x^3 + 1)$\\$(x^{18} + x^9 + 1)$} & 
   $d = d_{\text{BKLC}} - 1$ \\

   \hline 
   $[54,33,6]_2$ & $x^2 + x + 1$ & 
   \makecell{$x^{17} + x^{12} + x^{11} + x^3$} & 
   \makecell{$(x + 1)$\\$(x^{18} + x^9 + 1)$} & 
   $d = d_{\text{BKLC}} - 2$ \\

   \hline 
   $[62,30,10]_2$ & 
   \makecell{$(x+1)$\\$(x^5+x^2+1)$\\$(x^5+x^3+1)$} & 
\makecell{
$(x + 1)^2$\\
$(x^{15} + x^{14} + x^{13} + x^9 + x^8 + x^7$\\ $+ x^6 + x^5 + x^4 + x^3 + x^2 + x + 1)$
} & 
   \makecell{$(x + 1)$\\$(x^5 + x^3 + x^2 + x + 1)$\\$(x^5 + x^4 + x^2 + x + 1)$\\$(x^5 + x^4 + x^3 + x + 1)$\\$(x^5 + x^4 + x^3 + x^2 + 1)$} & 
   $d = d_{\text{BKLC}} - 2$ \\

   \hline 
   $[62,40,8]_2$ & $x + 1$ & 
 \makecell{
$x^3(x + 1)^2$\\
$(x^{14} + x^{13} + x^{12} + x^{11}$\\
$+\,x^{10} + x^9 + x^8 + x^7$\\
$+\,x^4 + x^3 + x^2 + x + 1)$
} & 
   \makecell{$(x + 1)$\\$(x^5 + x^3 + x^2 + x + 1)$\\$(x^5 + x^4 + x^2 + x + 1)$\\$(x^5 + x^4 + x^3 + x + 1)$\\$(x^5 + x^4 + x^3 + x^2 + 1)$} & 
   $d = d_{\text{BKLC}}$ \\

   \hline 
   $[66,30,12]_2$ & 
    \makecell{$(x + 1)$\\$(x^2 + x + 1)$\\$(x^{10} + x^7 + x^5 + x^3 + 1)$}& 
    \makecell{$x^4(x + 1)^6$\\$(x^2 + x + 1)$\\$(x^{10} + x^7 + x^5 + x^3 + 1)$}& 
   \makecell{$(x + 1)$\\$(x^2 + x + 1)$\\$(x^{10} + x^7 + x^5 + x^3 + 1)$\\$(x^{10} + x^9 + x^8 + x^7 + x^6 + x^5$\\$+\,x^4 + x^3 + x^2 + x + 1)$} & 
   $d = d_{\text{BKLC}} - 3$ \\
   \hline
\end{tabular}

\caption{Binary quasi-cyclic Euclidean LCD codes from Theorem \ref{main}, length $54 \le n \le 66.$}
\label{tab2}
\end{table}


We also  describe some examples of  ternary quasi-cyclic Euclidean LCD codes in Table \ref{tabternary}. We note that the codes with length $n=26$ and $n=28$ from \cite{grassltable} used for comparison are not LCD codes.
 
\begin{table}[H]   
    \centering
\tiny\begin{tabular}{|c|c|c|c|c|c|c|c|c|c|}
   \hline $C$  & $g_{11}(x)$& $g_{12}(x)$& $g_{22}(x)$& Remark  \\ 
   \hline $[14,7,6]_3$ & $x + 2$ & \makecell{$2x^5 + 2x^4 + x^3 + 2$} 
   & \makecell{$x^6 + x^5 + x^4 + x^3 + x^2 + x + 1$} 
   & Optimal, see \cite{grassltable} \\
\hline $[16,8,6]_3$ & $x + 1$ & \makecell{$2x^6 + x^5 + 2x^4 + 2x^2$} 
& \makecell{$(x + 2)(x^2 + 1)(x^2 + x + 2)(x^2 + 2x + 2)$} 
& Optimal, see \cite{grassltable} \\

\hline $[22,11,7]_3$ & $x + 2$ & \makecell{$x^9 + 2x^7 + 2x^6 + x^4 + 2x^2$} 
& \makecell{$(x^5 + 2x^3 + x^2 + 2x + 2)$\\$(x^5 + x^4 + 2x^3 + x^2 + 2)$} 
& $d=d_{\text{BKLCD}}$, see \cite{li2024} \\

\hline $[26,13,7]_3$ &\makecell{
$(x + 2)(x^3 + 2x + 2)$
}
& \makecell{$2x^8 + x^6 + 2x^3 + 2x^2$} 
& \makecell{$(x^3 + x^2 + 2)$\\$(x^3 + x^2 + x + 2)(x^3 + 2x^2 + 2x + 2)$} 
& $d = d_{\text{BKLC}}-1$, see \cite{grassltable} \\
\hline $[26,14,7]_3$ & $x^3 + 2x + 2$ & $x^7 + 2x^4 + x + 2$ 
& \makecell{$(x^3 + x^2 + 2)(x^3 + x^2 + x + 2)$\\$(x^3 + 2x^2 + 2x + 2)$} 
& $d = d_{\text{BKLC}}$, see \cite{grassltable} \\
\hline $[26,19,4]_3$ & $x + 2$ & $x^5 + 2x^4 + 2x^3 + x^2 + x + 1$ 
& \makecell{$(x^3 + x^2 + x + 2)$\\$(x^3 + 2x^2 + 2x + 2)$} 
& $d = d_{\text{BKLC}}-1$, see \cite{grassltable}\\
\hline $[28,14,8]_3$ & $x + 2$ & 
\makecell{
$x^3(x + 2)^4$\\
$(x^3 + 2x + 1)$
}
& \makecell{$(x + 2)(x^6 + x^5 + x^4 + x^3 + x^2 + x + 1)$\\$(x^6 + 2x^5 + x^4 + 2x^3 + x^2 + 2x + 1)$} 
& $d = d_{\text{BKLC}}-1$, see \cite{grassltable}\\

   \hline
    \end{tabular}

     \caption{Ternary quasi-cyclic Euclidean LCD codes from Theorem \ref{main}, length $14~\le~n~\le~28.$
    }
    \label{tabternary}
\end{table}





\section{One-generator quasi-cyclic codes of index 2}


In this section, we consider the   special case  when the code $C$ is generated by one element 
 $(g_{11}(x), g_{12}(x))$, where $g_{11}(x) \mid (x^m-1)$.  Let $g(x)=\gcd(g_{11}(x),g_{12}(x))$,
$g_{11}(x)=g(x)g_{11}'(x)$, $g_{12}(x)=g(x)g_{12}'(x)$. Let 
\[
g_{22}(x) = \dfrac{x^m-1}{g_{11}'(x)}.
\]
In view of Lemma \ref{lemma1gen}, we can assume that $C$ is   generated by two elements $(g_{11}(x), \tilde{g}_{12}(x))$ and $(0, g_{22}(x))$ satisfying Conditions  $(\ast)$, where 
  $\tilde{g}_{12}(x)=g_{12}(x) \mod{g_{22}(x)}$.

Furthermore, also by Lemma \ref{lemma1gen}, we have that  $\gcd(g_{11}(x),g_{22}(x))=g(x)$.  We note that the code $C$ can be generated by either the element  $(g_{11}(x),\tilde{g}_{12}(x))$ or the element  $(g_{11}(x),g_{12}(x))$. Hence without loss of generality, from now on we will write $g_{12}(x)$ instead of $\tilde{g}_{12}(x)$.

In the rest of this section, we still maintain the notation  in Section 3. 
We have
\[
\text{lcm}(g_{11},g_{22}) = \dfrac{g_{11}\cdot g_{22}}{\gcd(g_{11},g_{22})} = \dfrac{g \cdot g_{11}' \cdot (x^m-1)}{g \cdot g_{11}'} = x^m-1.
\]
Hence   $l(x)=1$ in this case.

\begin{lemma} \label{1a} If $g$ is self-reciprocal, then the following are equivalent.

1. $\gcd(t_{22},g_{12})=1$. 

2. $\gcd \left( t_{22},  g_{11}\bar{g}_{11}+g_{12} \bar{g}_{12}\right)=1.$
\end{lemma}
\begin{proof} Since $g$ is self-reciprocal, we have 
\begin{align*}
x^m-1 &= g \cdot  r_{11} \cdot t_{11} \cdot r_{22} \cdot t_{22} \\
    &= \alpha \cdot g \cdot r_{11} \cdot t_{11}^* \cdot r_{22} \cdot t_{22}^*,
\end{align*}
for some $\alpha \in F$. In particular, if   $h_j \mid t_{22}$, then $h_j \mid t_{11}^*$. 
We also note that for all $1 \le i \le s$, the irreducible polynomial $f_i$ is self-reciprocal, and so $f_i \nmid t_{22}$.  

1.  If $\gcd(t_{22},g_{12}) \ne 1$, then there exists $h_j$ such that $h_j \mid \gcd(t_{22},g_{12})$.   
Since $h_j \mid t_{22}$, we have that $h_j\mid t_{11}^* $. Then $h_j \mid g_{11}^*$ and so  $h_j \mid \bar{g}_{11}$. Then
$
h_j \mid  (g_{11}\bar{g}_{11}+g_{12} \bar{g}_{12}), 
$
and so $\gcd \left( t_{22},  g_{11}\bar{g}_{11}+g_{12} \bar{g}_{12}\right) \ne 1.$

2. Assume that $\gcd \left( t_{22},  g_{11}\bar{g}_{11}+g_{12} \bar{g}_{12}\right) \ne 1,$ that is    there exists $h_j$ such that 
\[ 
h_j \mid \gcd \left( t_{22},  g_{11}\bar{g}_{11}+g_{12} \bar{g}_{12}\right).
\]
Similar to part 1, the condition $h_j \mid t_{22}$ implies that $h_j \mid \bar{g}_{11}$. Then $h_j \mid g_{12} \bar{g}_{12}$. If $h_j  \mid\bar{g}_{12}$, then since $h_j \mid \bar{g}_{11}$, we also have that $h_j \mid \gcd(\bar{g}_{11}, \bar{g}_{12}) = \bar{g}$. Since $g$ is self-reciprocal, it follows that  $h_j \mid g$, and so $h_j \mid g_{12}$. This shows that $\gcd(t_{22},g_{12}) \ne 1$. \qedhere 
\end{proof}

\begin{lemma} \label{1b} If $g$ is self-reciprocal and $\gcd(t_{22},g_{12})=1$, then 
\[
\gcd \left(g_{11}',  g_{11}\bar{g}_{11}+g_{12} \bar{g}_{12}\right)=1.
\]
\end{lemma}
\begin{proof} 
Since $g$ is self-reciprocal, we have 
\begin{align*}
x^m-1 &= g \cdot  r_{11} \cdot t_{11} \cdot r_{22} \cdot t_{22} \\
    &= \alpha \cdot g \cdot r_{11} \cdot t_{11}^* \cdot r_{22} \cdot t_{22}^*,
\end{align*}
for some $\alpha \in F$. In particular, if   $h_j^* \nmid t_{22}$ for some $j$, then $h_j \nmid t_{11}$. Since $g=\gcd(g_{11},g_{12})$ and $g_{11}=g \cdot g_{11}'$, it follows that $\gcd( g_{11}', g_{12})=1$. 

1.  We  show that $\gcd( g_{11}', \bar{g}_{12})=1$.  Suppose that there exists an irreducible factor $a(x)$ of $x^m-1$ such that $a \mid \gcd( g_{11}', \bar{g}_{12})$. 
If $a=f_i$ for some $i$, then  since $f_i$ is self-reciprocal, $f_i \mid \bar{g}_{12}$ implies that  $f_i \mid g_{12}$.
On the other hand, $f_i \mid g_{11}'$ and $\gcd( g_{11}', g_{12})=1$ imply that $f_i \nmid g_{12}$, a contradiction. 

If $a=h_j$ for some $j$, then $h_j \mid \bar{g}_{12}$ implies that $h_j^* \mid g_{12}$. Since $\gcd(t_{22},g_{12})=1$, we have that $h_j^* \nmid t_{22}$. But then $h_j \nmid t_{11},$ implying  $h_j \nmid g_{11}',$ also a contradiction. Hence, $\gcd( g_{11}', \bar{g}_{12})=1$.

2.    The conditions $\gcd( g_{11}', g_{12})=1$ and $\gcd( g_{11}', \bar{g}_{12})=1$ imply that
\[
\gcd( g_{11}', g_{12}\bar{g}_{12})=1.
\]
Furthermore, since $g_{11}'\mid g_{11}$, we   obtain
\[
  \gcd \left(g_{11}',  g_{11}\bar{g}_{11}+g_{12} \bar{g}_{12}\right)=1. \qedhere
\]
\end{proof}

 \begin{theorem} \label{main2}  Let $C$ be a quasi-cyclic code generated by one element $(g_{11}(x),g_{12}(x))$, where $g_{11}(x)\mid (x^m-1)$. Let $g(x)=\gcd(g_{11}(x),g_{12}(x))$. Then $C$ is Euclidean LCD if and only if 
  \begin{equation} \label{eq1}
  \gcd \left(\dfrac{x^m-1}{g(x)},  g_{11}(x)\bar{g}_{11}(x)+g_{12}(x) \bar{g}_{12}(x)\right)=1.
  \end{equation}
 \end{theorem}  
\begin{proof} We recall from  Lemma \ref{lemma1gen}   we can assume that the code $C$ is generated by two elements 
$
(g_{11}(x),\tilde{g}_{12}(x))
$
and $(0,g_{22}(x))$ satisfying Conditions  $(\ast)$,  where
\[
g_{22}(x)= \dfrac{x^m-1}{g_{11}'(x)}.
\]
Furthermore, $\gcd(g_{11}(x),g_{22}(x))=g(x)$.
By Theorem \ref{main}, $C$ is Euclidean LCD if and only if conditions (I), (II), (III) and (IV) hold. Since $l(x)=1$, condition (II) holds trivially. 

1. We first show that if (I), (III) and (IV) hold, then \eqref{eq1} holds.
From Lemmata \ref{1a} and \ref{1b}, conditions (I) and (III) imply that 
\[
\gcd \left( g_{11}',  g_{11}\bar{g}_{11}+g_{12} \bar{g}_{12}\right)=1,
\]
and
\[
\gcd \left( t_{22},  g_{11}\bar{g}_{11}+g_{12} \bar{g}_{12}\right)=1.
\]
Together with conditions (IV), and the polynomials $g_{11}', r_{22}, t_{22}$ are pairwise relatively prime, we have that
\[
\gcd \left( g_{11}' r_{22} t_{22},  g_{11}\bar{g}_{11}+g_{12} \bar{g}_{12}\right)=1.
\]
This is condition \eqref{eq1}, since $x^m-1 = g \cdot g_{11}' \cdot r_{22} \cdot t_{22}$ under the assumption (I) that $g$ is self-reciprocal.

2. We now show that if \eqref{eq1} holds, then (I), (III) and (IV) hold. If $g$ is not self-reciprocal, then    there exists $h_j$ such that $h_j\mid g$ and  $h_j^*  \nmid g$. Since $h_j \mid g$, we have that $h_j^* \mid g^*$, which implies that $h_j^* \mid \bar{g}_{11}$ and $h_j^* \mid \bar{g}_{12}$.
Then
\[
h_j^* \mid   \gcd \left(\dfrac{x^m-1}{g},  g_{11}\bar{g}_{11}+g_{12} \bar{g}_{12}\right),
\]
contradicting \eqref{eq1}. Hence (I) holds. 
We then have 
$
x^m-1  = g   \cdot g_{11}' \cdot r_{22} \cdot t_{22}.
$
Now condition \eqref{eq1} becomes
\[
\gcd \left( g_{11}' \cdot r_{22} \cdot t_{22},  g_{11}\bar{g}_{11}+g_{12} \bar{g}_{12}\right)=1.
\]
Therefore,
\[
\gcd \left(  r_{22},  g_{11}\bar{g}_{11}+g_{12} \bar{g}_{12}\right)=1,
\]
which is condition (IV), and 
\[
\gcd \left( t_{22},  g_{11}\bar{g}_{11}+g_{12} \bar{g}_{12}\right)=1,
\]
which is equivalent to condition (III),  by Lemma \ref{1a}.
\end{proof}

This result is consistent with that in \cite{guneri2023} and \cite{guan2023}.  We also  describe some examples of  binary and ternary one-generator quasi-cyclic Euclidean LCD codes in Table \ref{tab1gen}.

\vspace{1 cm}
\begin{table}[H]   
    \centering
\tiny\begin{tabular}{|c|c|c|c|c|}
   \hline $C$  & $g_{11}(x)$ & $g_{12}(x)$ & Remark  \\ 

  \hline $[30,13,8]_2$ & $x^2 + x + 1$ & \makecell{$x^{14} + x^{10} + x^7 + x^6 + x^5 + x^4 + x + 1$} & Optimal LCD, see \cite{li2024}\\
    \hline $[30,14,8]_2$ & $x + 1$ & \makecell{$x^{14} + x^{13} + x^{12} + x^6 + x^5 + x^3 + x + 1$} & Optimal LCD, see \cite{li2024}\\
     \hline $[34,16,8]_2$ & $x + 1$ & \makecell{$x^{16} + x^{14} + x^{13} + x^{10} + x^9 + x^5 + x^2 + x$} & $d=d_{\text{BKLC}}$, see \cite{grassltable}\\ 
     
            \hline $[42,18,10]_2$ & $x^3 + 1$ & \makecell{$x^{19} + x^{18} + x^{17} + x^{16} + x^{12}$\\$+\,x^{11} + x^{10} + x^9 + x^3 + x$} & $d=d_{\text{BKLCD}}$, see \cite{wang2024} \\
     \hline $[42,19,10]_2$ & $x^2 + x + 1$ & \makecell{$x^{20} + x^{19} + x^{18} + x^{17} + x^{16}+ x^{12} $\\ $+ x^{11} + x^8 + x^7 + x^5 + x^3 +x^2 + x + 1$} & $d=d_{\text{BKLCD}}$, see \cite{wang2024}\\
         \hline $[42,20,10]_2$ & $x + 1$ & \makecell{$x^{19} + x^{18} + x^{17} + x^{13} + x^{12}$ \\ $+ x^9 + x^8 + x^7 + x + 1$} & $d=d_{\text{BKLCD}}$, see \cite{wang2024}\\

   \hline $[16,7,6]_3$ & $x + 1$ & \makecell{$2x^6 + x^5 + x^3 + 2x^2 + 2x$} & Optimal, see \cite{grassltable}\\
 
\hline $[20,9,7]_3$ & $x + 1$ & \makecell{$2x^8 + x^7 + 2x^6 + 2x^5 + x^3$} &  $d=d_{\text{BKLCD}}$, see \cite{li2024} \\

\hline $[22,10,8]_3$ & $x + 2$ & \makecell{$x^{10} + 2x^8 + x^7 + 2x^6 + x^5$\\$+\,2x^4 + 2x^2 + x$} &  Optimal LCD, see \cite{li2024} \\

   \hline
    \end{tabular}
   \caption{Binary and ternary one-generator quasi-cyclic Euclidean LCD codes from Theorem~\ref{main2}. 
    }
    \label{tab1gen}
\end{table}

\section{Symplectic LCD quasi-cyclic codes of index 2}
Assume that all the notations are the same as in the previous sections. First, we recall some definitions and results, for details see \cite[Section V]{ezerman2025}. 
 For $x=(x_1|x_2)$ and $y=(y_1|y_2)$ in $F^{2m}$, where $x_i,y_i\in F^m$ for $i=1,2$, we have 
 
 $$\langle x,y\rangle_s=\langle x_1,y_2\rangle_e-\langle x_2,y_1\rangle_e.$$
 Define $\tau : F^{2m}\to F^{2m}$ as $(x_1|x_2)\mapsto (x_2|-x_1)$, where $x_1,x_2\in F^m$. Then, we have 
 
 $$\langle (x_1|x_2), (y_1|y_2)\rangle_s=-\langle \tau ((x_1|x_2)),(y_1|y_2)\rangle_e. $$
 
 From the above relation, it is easy to see that the symplectic dual $C^{\perp_s}$ of a QC code $C$ of length $2m$ and index $2$ satisfies $$C^{\perp_s}=\tau(C^{\perp_e})=\tau(C)^{\perp_e}.$$

 Let $C$ be a QC code of length $2m$ and index $2$ with CRT decomposition given in \eqref{crtC}. By  extending the map $\tau$  canonically to the vector spaces $(F_i)^2,(H_j')^2 , (H_j'')^2$ and applying the maps component-wise to \eqref{crtCdual}, we obtain  symplectic dual of $C$ (see \cite[Proposition V.1]{ezerman2025}).
 \begin{proposition}\label{slcdconst}
 	Let $C$ be a QC code with CRT decomposition as given in \eqref{crtC}. Then its symplectic dual $C^{\perp_s}$ is given by \begin{equation} \label{crtCsdual}
 		C^{\perp_s} = \left(\bigoplus_{i=1}^{s} C_i^{\perp_{s_i}} \right) 
 		\oplus
 		\left(\bigoplus_{j=1}^{p} \left(  C_j''^{\perp_s}
 		\oplus   C_j'^{\perp_s}   \right) \right),
 	\end{equation} 
 	where ${C^{\perp_{s_i}}}=\tau(C_i)^{\perp_h}$ (see Definition \eqref{hinnerproduct} for $\perp_h$) for each $1\le i \le s$ and $ \perp_s$ denotes the usual symplectic dual on $(H_j')^2 \cong (H_j'')^2$ for all $1\leq j\leq p$.
 \end{proposition}
Using the above characterization, we have the following characterization of symplectic LCD QC codes in terms of constituents (see \cite[Eq. V.7]{ezerman2025}). 
\begin{theorem}
		Let $C$ be a QC code with CRT decomposition as given in \eqref{crtC}. Then $C$ is symplectic LCD code if and only if $C\cap C^{\perp_{s_i}}=\{\mathbf{0}\}$ for all $1\leq i\leq s$, and $C_j'\cap C_j''^{\perp_s}=\{\mathbf{0}\}=C_j'^{\perp_s}\cap C_j''$ for all $1\leq j\leq p$.
\end{theorem}
Now, we give a polynomial characterization of symplectic LCD QC codes of index $2$ and for one-generator QC codes.

\begin{theorem} \label{slcdmain} 
Let $C$ be a quasi-cyclic code   of index $2$. Let $(g_{11}(x), g_{12}(x))$ and $(0, g_{22}(x))$ be the generators of $C$ satisfying Conditions $(\ast)$.   
		Then $C$ is symplectic LCD if and only if all of the following conditions are true:
		\begin{enumerate}[(I)]
			\item  $g$ is self-reciprocal.
			\item  $l$ is self-reciprocal. 
			\item  $r_{11}(x)=1$. 
			\item $\gcd(r_{22}(x), g_{11}(x)\bar{g}_{12}(x)-g_{12}(x) \bar{g}_{11}(x))=1$.
		\end{enumerate}
	\end{theorem}

	The proof will follow in manner similar to that of Theorem \ref{main}. For instance, Conditions $(I)$ and $(II)$ will follow similar to Lemmata \ref{nI} and \ref{nII} using dimension arguments. We prove the necessary Conditions $(III)$ and $(IV)$ in the following lemmata. The sufficient part follows using arguments similar to those of Lemmata  \ref{s1} and \ref{s2}.

\begin{lemma}\label{nIIIs}
Assume that $(I)$ and $(II)$ hold. If $(III)$ does not hold, then $C$ is not symplectic LCD.
\end{lemma}
\begin{proof}
	Since  $g$ and $l$ are self-reciprocal, 
	\[
	x^m-1  = g \cdot l \cdot r_{11} \cdot t_{11} \cdot r_{22} \cdot t_{22}.  
	\] 
	Furthermore, the  reciprocal of  $x^m-1$ is  $-(x^m-1)$, so we can rewrite 
	\begin{align*}
		x^m-1   &= \alpha \cdot g \cdot l \cdot r_{11} \cdot t_{11}^*    \cdot  r_{22} \cdot  t_{22}^*,
	\end{align*} 
	for some $\alpha 
	\in F$. Assume that $(III)$ is not true, that is, there exists an irreducible  $a(x)$ such that $a(x)$ divide $r_{11}$.\\
	1. If $a(x)$ is self-reciprocal, then $a(x)=f_i(x)$ for some $1\leq i\leq s$. As $f_i(x)$ divides $r_{11}$, we have $f_i(x)\mid g_{11}(x)$ and $f_i(x)\nmid g_{22}(x)$. Consequently, $G_i$ is of the form
	\[
	G_i= \begin{bmatrix}
		0 & g_{12}(\xi^{u_i}) \\
		0 & g_{22}(\xi^{u_i}) 
	\end{bmatrix},
	\]
where 	$g_{22}(\xi^{u_i})\neq 0$. Hence $C_i=\langle (0,1)\rangle=C_i^{\perp_{s_i}}$, that is, $C_i\cap C_i^{\perp_{s_i}}\neq\{\mathbf{0}\}$. Thus, $C$ is not symplectic LCD.\\
2. If $a(x)$ is not self-reciprocal, then $a(x)=h_j(x)$ for some $1\leq j\leq p$. As $h_j(x)$ divides $r_{11}$ and $r_{11}$ is self-reciprocal, therefore $h_j^*(x)$ divides $r_{11}$. It follows that $h_j(x)\mid g_{11}(x)$, $h_j(x)\nmid g_{22}(x)$,  $h_j^*(x)\mid g_{11}(x)$ and  $h_j^*(x)\nmid g_{22}(x)$. Then $G_j'$ and $G_j''$ are of the form
$$G_j'=   \begin{bmatrix}
	0 & g_{12}(\xi^{v_j}) \\
	0 & g_{22}(\xi^{v_j}) 
\end{bmatrix} ,
G_j''=   \begin{bmatrix}
	0 & \bar{g}_{12}(\xi^{v_j}) \\
	0 & \bar{g}_{22}(\xi^{v_j}) 
\end{bmatrix},$$
  where	$g_{22}(\xi^{v_j})\neq 0$ and $\bar{g}_{22}(\xi^{v_j})\neq 0$. Hence $C_j'=\langle (0,1)\rangle=C_j'^{\perp_{s}}$ and $C_j''=\langle (0,1)\rangle=C_j''^{\perp_{s}}$. It follows that $C_j'\cap C_j''^{\perp_{s}}\neq \{\mathbf{0}\}\neq C_j'^{\perp_{s}}\cap C_j''$. Thus, $C$ is not symplectic LCD.
\end{proof}

\begin{lemma}
	Assume $(I)$ and $(II)$ hold. If $(IV)$ is not true, then $C$ is not symplectic LCD.
\end{lemma}
\begin{proof}
Similar to the proof of Lemma \ref{nIIIs}, if $g$ and $l$ are self-reciprocal, we have
\[
x^m-1  = g \cdot l \cdot r_{11} \cdot t_{11} \cdot r_{22} \cdot t_{22}.
\] 
Assume that condition (IV) is not true, that is, there exists an irreducible factor  $a(x)$ of $x^m-1$ such that 
$$a \mid \gcd(r_{22}, g_{11}\bar{g}_{12}-g_{12} \bar{g}_{11}).$$
Since $a \mid r_{22}$, we have that  $a \nmid g_{11}$,  $a \mid g_{22}$,  $a \nmid g_{11}^*$,  and $a \mid g_{22}^*$.   We have two cases depending on whether $a$ is self-reciprocal. 

1. $a$ is self-reciprocal, that is, $a=f_i$ for some $i$. Since $f_i  \nmid g_{11}$ and $f_i \mid g_{22}$, $G_i$ is of the form
\[
G_i=  \begin{bmatrix}
	g_{11}(\xi^{u_i}) & g_{12}(\xi^{u_i}) \\
	0 & 0
\end{bmatrix},
\] 
where $g_{11}(\xi^{u_i}) \ne 0$. Then   $C_i =\langle (g_{11}(\xi^{u_i}), g_{12}(\xi^{u_i}))\rangle$ is a 1-generator code, whose dual is $C_i^{\perp_{s_i}} =\langle (\bar{g}_{11}(\xi^{u_i}),\bar{g}_{12}(\xi^{u_i}))\rangle$ (by definition of $\perp_{s_i}$, see Proposition \ref{slcdconst}).   But since $f_i \mid  g_{11}\bar{g}_{12}-g_{12} \bar{g}_{11}$, 
\[
g_{11}(\xi^{u_i})\bar{g}_{12}(\xi^{u_i})-g_{12}(\xi^{u_i}) \bar{g}_{11}(\xi^{u_i}) =0, 
\]
and hence $C_i = C_i^{\perp_{s_i}}$. Hence $C_i \cap  C_i^{\perp_{s_i}} \ne \{\mathbf{0}\}$, and  $C$ is not symplectic LCD. 

2. $a$ is not self-reciprocal, that is, $a=h_j$ for some $j$. With the same reasoning as in case 1, we have that $C_j' =\langle (g_{11}(\xi^{v_j}), g_{12}(\xi^{v_j}))\rangle$ and $C_j'^{\perp_s} =\langle (g_{11}(\xi^{v_j}), g_{12}(\xi^{v_j}))\rangle$. Also, $C_j'' =\langle (\bar{g}_{11}(\xi^{v_j}), \bar{g}_{12}(\xi^{v_j}))\rangle$. 
But since $h_j \mid g_{11}\bar{g}_{12}-g_{12} \bar{g}_{11}$, 
\[
g_{11}(\xi^{v_j})\bar{g}_{12}(\xi^{v_j})-g_{12}(\xi^{v_j}) \bar{g}_{11}(\xi^{v_j}) =0, 
\]
and so $C_j'' = C_j'^{\perp_s}$. Hence $C_j'' \cap  C_j'^{\perp_s} \ne \{\mathbf{0}\}$, and hence $C$ is not symplectic LCD. 
\end{proof}


Next, we give the characterization for one-generator QC codes of index $2$. The proof is similar to that of Theorem \ref{main2} using  Theorem \ref{slcdmain} above. Therefore, we omit the proof.
\begin{theorem} \label{slcd1}  
Let $C$ be a quasi-cyclic code generated by one element $(g_{11}(x),g_{12}(x))$, where $g_{11}(x)\mid (x^m-1)$. Let $g(x)=\gcd(g_{11}(x),g_{12}(x))$. Then $C$ is symplectic LCD if and only if 
	\begin{equation} 
		\gcd \left(\dfrac{x^m-1}{g(x)},  g_{11}(x)\bar{g}_{12}(x)-g_{12}(x) \bar{g}_{11}(x)\right)=1.
	\end{equation}
\end{theorem} 

This result is consistent with that in  \cite{guan2023}. 
Restricting to the case $r_{11}(x)=r_{22}(x)=1$ in Theorem \ref{slcdmain}, we can get the following sufficient condition for a code to be symplectic LCD. 
\begin{corollary} \label{cor-sympl} 
Let $C$ be a quasi-cyclic code  of index $2$. Let $(g_{11}(x), g_{12}(x))$ and $(0, g_{22}(x))$ be the generators of $C$ satisfying Conditions $(\ast)$. If the polynomial $g(x) = \gcd(g_{11}(x),g_{22}(x)) = \gcd(g_{11}(x),g_{11}^*(x)) = \gcd(g_{22}(x),g_{22}^*(x))$ is self-reciprocal and $g_{11}(x)g_{22}(x)$ is self-reciprocal, then $C$ is symplectic LCD.    
\end{corollary}



We now describe some examples of binary quasi-cyclic symplectic LCD codes using Theorems \ref{slcdmain}, \ref{slcd1} and Corollary~\ref{cor-sympl}.  
In \cite{calderbank1998}, Calderbank  et al. proved that the binary symplectic inner product and the quaternary trace-Hermitian inner product are equivalent. Therefore, a binary symplectic LCD code with parameters $[2n, k, d_s]_{2}$ is also a quaternary additive complementary dual (ACD) code with parameters $(n, k/2, d_s)_{4}$.

We note that it is still a challenging problem to efficiently construct quaternary additive codes that outperform linear codes, see \cite{guan2023b}. In Table \ref{tabadditive}, we describe quatenary ACD codes
which match the performance of Hermitian LCD codes in \cite{araya2024} and ACD codes in \cite{guan2023b}.

\vspace{1 cm}
\begin{table}[h!]   
    \centering
\tiny\begin{tabular}{|c|c|c|c|c|}
   \hline ACD code & $g_{11}(x)$ & $g_{12}(x)$ & $g_{22}(x)$ & Comparison\\ 

\hline $(15,5,8)_4$ & \makecell{$(x + 1)(x^4 + x^3 + x^2$\\$+\,x + 1)$} & \makecell{$x(x + 1)(x^4 + x^3 + x^2 + x + 1)$\\$(x^3 + x + 1)$} & $x^{15} - 1$ & $[15,5,8]_4^H$, see \cite{araya2024}\\
\hline $(15,6,7)_4$ & \makecell{$x^4 +x^3 +x^2 +x+1$} & 
\makecell{
$(x + 1)(x^3 + x^2 + 1)$\\
$(x^4 + x^3 + x^2 + x + 1)$
}
& \makecell{$(x^4 +x^3 +x^2 +x+1)(x^4 +x^3 +1)$\\$(x^4 +x+1)(x^2+x+1)$} & $[15,6,7]_4^H$, see \cite{araya2024}\\    
\hline $(15,9,5)_4$ & $1$ & $(x + 1)^2 (x^4 + x + 1)$ & \makecell{$(x^4 + x + 1)(x^4 + x^3 + 1)$\\$(x^4 + x^3 + x^2 + x + 1)$} & $[15,9,5]_4^H$, see \cite{araya2024} \\

\hline $(19,9,8)_4$ & $x + 1$ &
\makecell{
$(x + 1)^3(x^4 + x^3 + x^2 + x + 1)$\\
$(x^{11} + x^6 + x^5 + x^4 + 1)$
}&
$x^{19} - 1$ & $[19,9,8]_4^H$, see \cite{araya2024} \\

\hline $(19,10,7)_4$ & $1$ & \makecell{$x^{11} + x^8 + x^7 + x^6$\\$+\,x^4 + x^3 + x^2 + x$} & \makecell{$x^{18} + x^{17} + x^{16} + x^{15} + x^{14} + x^{13}$\\$+\,x^{12} + x^{11} + x^{10} + x^9 + x^8 + x^7$\\$+\,x^6 + x^5 + x^4 + x^3 + x^2 + x + 1$} &  $[19,10,7]_4^H$, see \cite{araya2024} \\

\hline $(25,13,8)_4$ & $1$ & \makecell{$x^{23} + x^{12} + x^{11} + x^{10}$\\$+\,x^9 + x^5 + x$} & \makecell{$(x^4 + x^3 + x^2 + x + 1)$\\$(x^{20} + x^{15} + x^{10} + x^5 + 1)$} & \makecell{$[25,13,8]_4^H$,   see \cite{araya2024}, \cite{guan2023b}}\\

\hline $(27,13,9)_4$ & $x^2 + x + 1$ & \makecell{
$x^2(x+1)^3(x^2 + x + 1)^3$\\$(x^7   + x + 1)$ 
} & 
\makecell{$(x^2 + x + 1)$\\$(x^6 + x^3 + 1)(x^{18} + x^9 + 1)$} & $[27,13,9]_4^H$, see \cite{araya2024} \\


   \hline
    \end{tabular}

    \caption{Additive   complementary dual codes from Theorems \ref{slcdmain}, \ref{slcd1} and Corollary~\ref{cor-sympl}.}
    \label{tabadditive}
\end{table}

\section{Hermitian LCD quasi-cyclic codes of index 2}
In this section, we consider the finite field $\mathbb{F}_{q^2}$, where $q$ is a prime power. Let $R=\mathbb{F}_{q^2}[x]/\langle x^m-1\rangle$. To characterize QC Hermitian LCD codes, we decompose the ring $R$ (subsequently a QC code $C$) slightly differently from the Euclidean case by factoring $(x^m-1)$ into self-conjugate-reciprocal polynomials in $\mathbb{F}_{q^2}[x]$. For details, see \cite{ezerman2025,lv2019,sangwisut2016}.

Recall that the  \textit{conjugate} of a polynomial $f(x)=f_0+f_1x+\dots+f_{k}x^{k}\in \mathbb{F}_{q^2}[x]$ is defined as 

$$f^{[q]}(x)=f_0^q+f_1^qx+\dots+f_{k}^qx^{k}$$ and  \textit{conjugate-reciprocal} is defined as
$$f^{\dagger}(x)={f^*}^{[q]}(x)=x^{\deg f(x)}f^{[q]}(x^{-1}).$$

Note that $(f^{\dagger})^{\dagger}(x)=f(x)$. We say a polynomial is  \textit{self-conjugate-reciprocal} if $f^{\dagger}(x)=\alpha f(x)$ for some $\alpha\in  \mathbb{F}_{q^2}$. Let $f(x) = a_0+a_1x+a_2x+ \dots + a_{m} x^{m}\in \mathbb{F}_{q^2}[x]$.
The   \textit{conjugate transpose polynomial} of $f(x)$ is the polynomial
\[
\hat{f} (x) =x^{m} f^{[q]}(x^{-1}) = a^q_m+a^q_{m-1}x+a^q_{m-2}x+ \dots + a^q_0 x^m \in \mathbb{F}_{q^2}[x].
\]
Then $\hat{f}(x)=x^{m-\deg f(x)}f^{\dagger}(x)$. 
  
Assume that $\gcd(q,m)=1$. We factor $(x^m-1)$ into distinct irreducible polynomials in $\mathbb{F}_{q^2}$ as follows 

$$
x^m-1= \delta \prod_{i=1}^sf_i(x) \prod_{j=1}^p h_j(x)h^{\dagger}_j(x),
$$
where $\delta$ is nonzero in $\mathbb{F}_{q^2}$, $f_i(x)$ is self-conjugate-reciprocal for all $1\le i \le s$, and  $h_j(x), h_j^{\dagger}(x)$ are conjugate-reciprocal  pairs for all $1 \le j \le p$.

For each $i$ and $j$, let $F_i=\mathbb{F}_{q^2}[x]/ (f_i(x) )$, 
$H_j'= \mathbb{F}_{q^2}[x]/ (h_j(x) )$,   and
$H_j''= \mathbb{F}_{q^2}[x]/ (h_j^{\dagger}(x) )$. 
Let $\xi$ be a primitive $m^{\text{th}}$ root of unity.
Let $\xi^{u_i}$ and  $\xi^{v_j}$ be roots of $f_i(x)$ and $h_j(x)$, respectively.
Then $h_j^{\dagger}(\xi^{-qv_j})=0$,
$F_i \cong \mathbb{F}_{q^2}(\xi^{u_i})$,  $H_j' \cong \mathbb{F}_{q^2}(\xi^{v_j})$,  and 
$H_j'' \cong \mathbb{F}_{q^2}(\xi^{-qv_j}) = \mathbb{F}_{q^2}(\xi^{v_j})$. 

The map $\hat\ : f(x) \mapsto \hat{f}(x)$ can be naturally extended to the following isomorphisms:   

\begin{equation}
	\begin{split}
		\hat\  :\  \mathbb{F}_{q^2}[x]/ (f_i(x) )&\to \mathbb{F}_{q^2}[x]/ (f_i (x)) ,\\
		\hat\  :\  \mathbb{F}_{q^2}[x]/ (h_j(x) )&\to \mathbb{F}_{q^2}[x]/ (h_j^{\dagger}(x) ) .\\
	\end{split}
\end{equation} 
(Here $\hat\  : a(x) + (h_j(x) )\mapsto \hat{a}(x) + (h_j^{\dagger}(x) )$.)  
Therefore, the map $\hat\ $ is an isomorphism from $H_j'=\mathbb{F}_{q^2}[x]/ ( h_j(x) )$ to 
$H_j''=\mathbb{F}_{q^2}[x]/ ( h_j^\dagger(x) )$. 

Define isomorphisms $\mu_j : H'_j \to \mathbb{F}_{q^2}(\xi^{v_j})$ by $\mu_j\big(a(x) + (h_j)\big) =a(\xi^{v_j})$ and 
$\nu_j : H''_j \to \mathbb{F}_{q^2}(\xi^{v_j})$ by $\nu_j\big(a(x) + (h_j^{\dagger})\big) =\hat{a}(\xi^{v_j})$. 
Then the following diagram is commutative:
\begin{center}
\begin{tikzcd}
        H'_j \arrow[r, "\mu_j"] \arrow[d, "\widehat\ "']
        & \mathbb{F}_{q^2}(\xi^{v_j}) \\
        H''_j \arrow[ru, "\nu_j",rightarrow] 
    \end{tikzcd}
\end{center}
Therefore, isomorphisms $\mu_j$ and $\nu_j$ allow us to identify $H'_j$ and $H''_j$ with the field $\mathbb{F}_{q^2}(\xi^{v_j})$.

By the Chinese Remainder Theorem (CRT), $R$ can be decomposed as
$$
R \cong \left(\bigoplus_{i=1}^{s} F_i \right) 
\oplus
\left(\bigoplus_{j=1}^{p} \left( H_j'
\oplus  H_j'' \right) \right).
$$ 
In this setup, the isomorphism between $R$ and its CRT decomposition is given by 
\[
a(x) \mapsto \left(a(\xi^{u_1}), \dots, a(\xi^{u_s}),a(\xi^{v_1}),\hat{a}(\xi^{v_1}), \dots, a(\xi^{v_p}),\hat{a}(\xi^{v_p})\right)). 
\]

This isomorphism extends naturally to  $R^\ell$, which implies that
$$
R^\ell \cong \left(\bigoplus_{i=1}^{s} F_i^\ell \right) 
\oplus
\left(\bigoplus_{j=1}^{p} \left( (H_j')^\ell
\oplus  (H_j'')^\ell  \right) \right).
$$ 
Then, a QC code $C$ of index $\ell$ can be  decomposed as 
\begin{equation} \label{crtCH}
	C \cong \left(\bigoplus_{i=1}^{s} C_i \right) 
	\oplus
	\left(\bigoplus_{j=1}^{p} \left(  C_j' 
	\oplus   C_j''   \right) \right), 
\end{equation} 
where each component codes $C_i$, $C_j'$ and $C_j''$ are constituent linear codes of length $\ell$ over the fields $F_i$,  $H_j'$, and $H_j''$, respectively.

In this setup, the constituents are described as: if $ C$ is an $r$-generator QC code with generators
\[
\{ (a_{1,1}(x),\dots ,a_{1,\ell}(x)) ,\dots, (a_{r,1}(x),\dots,a_{r,\ell} (x) ) \} \subset R^\ell,
\]
then 
\begin{align*}
	C_i &= \text{Span}_{F_i} \{ (a_{k,1}(\xi^{u_i}),\dots ,a_{k,\ell}(\xi^{u_i})) : 1 \le k \le r \}, \text{ for } 1 \le i \le s, \\
	C_j' &= \text{Span}_{H_j'} \{ (a_{k,1}(\xi^{v_j}),\dots ,a_{k,\ell}(\xi^{v_j})) : 1 \le k \le r \}, \text{ for } 1 \le j \le p, \\
	C_j'' &= \text{Span}_{H_j''} \{ (\hat{a}_{k,1}(\xi^{v_j}),\dots ,\hat{a}_{k,\ell}(\xi^{v_j})) : 1 \le k \le r \}, \text{ for } 1 \le j \le p.
\end{align*}

\begin{remark} 
The constituent code $C_j'$ over $H_j'$ is obtained by evaluating $a(x)$ at $\xi^{v_j}$, while the constituent  code $C_j''$ over $H_j''$ is obtained by evaluating $\hat{a}(x)$ at $\xi^{v_j}$. 
In \cite{ezerman2025}, the authors utilized a slightly different isomorphism. 
They defined the constituent code $C_j'$ over $H_j'$  by evaluating $a(x)$ at $\xi^{v_j}$ and the constituent  code $C_j''$ over $H_j''$  by evaluating $a(x)$ at $\xi^{-qv_j}$, which means 
$a(\xi^{-qv_j}) = (\hat{a}(\xi^{v_j} ))^q$ 
(see \cite[Eq. (IV.10)]{ezerman2025} and remark after Eq. (IV.14)).
\end{remark}
In this set up, we have the following characterization of the Hermitian dual (see \cite{ezerman2025,lv2019}).

\begin{theorem}\label{crtCdualh}
Let $C$ be a QC code with its CRT decomposition given in \eqref{crtCH}, the Hermitian dual of $C$ is given by
\begin{equation} 
	C^{\perp_h} = \left(\bigoplus_{i=1}^{s} C_i^{\perp_H} \right) 
	\oplus
	\left(\bigoplus_{j=1}^{p} \left(  C_j''^{\perp_e}
	\oplus   C_j'^{\perp_e}   \right) \right),
\end{equation} 
where, $\perp_e$ is the Euclidean dual on $(H_j')^\ell \cong (H_j'')^\ell$ for $1\leq j\leq p$ and  ${\perp_H}$ denotes the dual on $F_i^\ell$ (for each $1\le i \le s$) with respect to the following inner product 
\begin{equation}
	\langle \textbf{c},\textbf{d} \rangle_H= \sum_{k=1}^\ell c_k(\xi^{u_i})\hat{d}_k(\xi^{u_i}),
\end{equation}
for all
$
\textbf{c} = (c_1(\xi^{u_i}), \dots,  c_\ell(\xi^{u_i})), \textbf{d} = (d_1(\xi^{u_i}), \dots,  d_\ell(\xi^{u_i})) \in F_i^\ell.
$ 	
\end{theorem}
We have the following characterization of Hermitian LCD QC codes in terms of constituents (see \cite[Eq. (IV.16)]{ezerman2025}).

\begin{theorem} \label{qclcdH} Let $C$ be a $q$-ary QC code of length $m\ell$ and index $\ell$ with a CRT decomposition as in \eqref{crtCH}.
	Then $C$ is Hermitian LCD if and only if 
	$C_i \cap C_i^{\perp_H} = \{\mathbf{0}\}$ for all $1 \le i \le s$, and $C_j' \cap C_j''^{\perp_e} = \{\mathbf{0}\}$, $C_j'' \cap C_j'^{\perp_e} = \{\mathbf{0}\}$ 
for all $1\le j \le p$. 
\end{theorem}

Let $C$ be a QC code of length $2m$ and index $2$ over  $\mathbb{F}_{q^2}$, generated by $(g_{11}(x),g_{12}(x))$ and $(0,g_{22}(x))$ satisfying Conditions $(\ast)$. Similar to Section 3, in the above setting, each constituent of $C$ is generated by the rows of a $2 \times 2$ matrix over its field of definition. Explicitly, $C_i, C_j'$ and $C_j''$ are generated by the rows of the matrices
$$
G_i=  \begin{bmatrix}
	g_{11}(\xi^{u_i}) & g_{12}(\xi^{u_i}) \\
	0 & g_{22}(\xi^{u_i}) 
\end{bmatrix} ,
G_j'=   \begin{bmatrix}
	g_{11}(\xi^{v_j}) & g_{12}(\xi^{v_j}) \\
	0 & g_{22}(\xi^{v_j}) 
\end{bmatrix}  ,
G_j''=   \begin{bmatrix}
	\hat{g}_{11}(\xi^{v_j}) & \hat{g}_{12}(\xi^{v_j}) \\
	0 & \hat{g}_{22}(\xi^{v_j}) 
\end{bmatrix}  ,
$$
respectively.

Similar to  the Euclidean case, we prepare the background for the Hermitian case. Let $g(x)= \gcd(g_{11}(x),g_{22}(x))$. Since we assume $\gcd(q,m)=1$, the condition 
$$g_{11}(x)g_{22}(x) \mid (x^m-1)g_{12}(x)$$ 
in Theorem \ref{qc2} is equivalent to the condition $g(x) \mid g_{12}(x)$, see Remark \ref{remarkgcd}.

Let  $l(x)= (x^m-1)/\text{lcm}(g_{11}(x),g_{22}(x))$. Let $g_{11}(x)=g(x)   g_{11}'(x), g_{22}=g(x) g_{22}'(x)$, and
\begin{align*}
	g_{11}'(x)=r_{11}(x) t_{11}(x),\\
	g_{22}'(x)=r_{22}(x) t_{22}(x),
\end{align*}
where $r_{11}(x)=\gcd(g'_{11}(x), g'^{\dagger}_{11}(x))$, and $r_{22}(x)=\gcd(g'_{22}(x), g'^{\dagger}_{22}(x))$. Then $r_{11}(x)$ and $r_{22}(x)$ are self-conjugate-reciprocal.

Now, we provide a polynomial characterization of QC Hermitian LCD codes. The proof is similar to that of Theorem \ref{main} and Theorem \ref{main2} (by replacing $\bar{a}(x)$ with $\hat{a}(x)$ and $*$ with $\dagger$). Therefore we omit the proof.  

\begin{theorem} \label{qch} Let $C$ be a quasi-cyclic code   of index $2$. Let $(g_{11}(x), g_{12}(x))$ and $(0, g_{22}(x))$ be the generators of $C$ satisfying Conditions $(\ast)$.   
	Then $C$ is Hermitian LCD if and only if all of the following conditions are true:
	\begin{enumerate}[(I)]
		\item  $g$ is self-conjugate-reciprocal.
		\item  $l$ is self-conjugate-reciprocal. 
		\item  $\gcd(t_{22}(x),g_{12}(x))=1$. 
		\item $\gcd(r_{22}(x), g_{11}(x)\hat{g}_{11}(x)+g_{12}(x) \hat{g}_{12}(x))=1$.
	\end{enumerate}
\end{theorem}

\begin{theorem} \label{Hlcd1}  Let $C$ be a quasi-cyclic code generated by one element $(g_{11}(x),g_{12}(x))$, where $g_{11}(x)\mid x^m-1$. Let $g(x)=\gcd(g_{11}(x),g_{12}(x))$. Then $C$ is Hermitian LCD if and only if 
	\begin{equation*} 
		\gcd \left(\dfrac{x^m-1}{g(x)},  g_{11}(x)\hat{g}_{11}(x)+g_{12}(x) \hat{g}_{12}(x)\right)=1.
	\end{equation*}
\end{theorem}

In \cite{carlet2018}, Carlet et al. showed that any linear code is equivalent to a Hermitian LCD code for $q > 4$. This motivates us to search for quaternary Hermitian LCD codes with good parameters.  
In Tables \ref{tabhermitian1} and \ref{tabhermitian2}, we provide several examples of  quasi-cyclic Hermitian LCD codes over $\mathbb{F}_4=\mathbb{F}_2[w]$ with $w^2+w+1=0$, using Theorem \ref{qch}. 

In Table \ref{tabhermitian1}, we use the abbreviation \textbf{BKLC} to indicate that a code is the \emph{best known linear code} according to \cite{grassltable}. Similarly, \textbf{BKHLCD} denotes the \emph{best known Hermitian LCD code}, according to \cite{araya2024} and \cite{sok2020}.

In Table~\ref{tabhermitian2}, we present examples of Hermitian LCD codes with minimum distance larger than previously known codes in the literature~\cite{crnkovic2023, grbac} with the same length and dimension. The codes reported in~\cite{crnkovic2023, grbac} were used as benchmarks for ACD codes in~\cite{zhu2024}.

\begin{table}[H]
    \centering
\scriptsize
\begin{tabular}{|c|c|c|c|c|}
    \hline 
    $C$  & $g_{11}(x)$ & $g_{12}(x)$ & $g_{22}(x)$ & Remark \\
    \hline
    $[14,8,5]_4^H$ & $1$ & 
    \makecell{$wx^5 + w^2x^4 + wx^3$\\$+\,x^2 + x + 1$} &
    \makecell{$(x^3 + x + 1)$\\$(x^3 + x^2 + 1)$} &
    $d = d_{\text{BKHLCD}}$, see \cite{araya2024} \\
    \hline

    $[22,11,8]_4^H$ & $x+1$ &
    \makecell{$(x+1)$\\$(w^2x + wx^2+\,x^4 + wx^6)$} &
    \makecell{$(x^5 + wx^4 + x^3 + x^2 + w^2x + 1)$\\$(x^5 + w^2x^4 + x^3 + x^2 + wx + 1)$} &
    $d = d_{\text{BKHLCD}}$, see \cite{araya2024} \\
    \hline

    $[22,12,7]_4^H$ & $1$ &
    \makecell{$w^2x^7 + w^2x^5 + w^2x^4$\\$+\,w^2x^2 + w^2x + 1$} &
    \makecell{$(x^5 + wx^4 + x^3 + x^2 + w^2x + 1)$\\$(x^5 + w^2x^4 + x^3 + x^2 + wx + 1)$} &
    $d = d_{\text{BKHLCD}}$, see \cite{araya2024} \\
    \hline

    $[22,16,4]_4^H$ & $x+1$ &
    \makecell{$x^5 + wx^4$\\$+\,x^2 + x + 1$} &
    \makecell{$x^5 + wx^4 + x^3 + x^2 + w^2x + 1$} &
    $d = d_{\text{BKHLCD}}$, see \cite{araya2024} \\
    \hline

    $[22,17,4]_4^H$ & $1$ &
    \makecell{$w^2x^4 + x^2 + x + w$} &
    \makecell{$x^5 + wx^4 + x^3 + x^2 + w^2x + 1$} &
    $d = d_{\text{BKHLCD}}$, see \cite{araya2024} \\
    \hline

    $[26,13,8]_4^H$ & $x+1$ &
    \makecell{$w^2x^9 + x^8 + w^2x^7$\\$+\,x^6 + x^4 + wx^3$\\$+\,x^2 + wx$} &
    \makecell{$(x^6 + wx^5 + w^2x^3 + wx + 1)$\\$(x^6 + w^2x^5 + wx^3 + w^2x + 1)$} &
    $d = d_{\text{BKHLCD}}$, see \cite{araya2024} \\
    \hline

    $[34,16,11]_4^H$ & 
    \makecell{$(x + 1)(x^4 + wx^3$\\$+\,x^2 + wx + 1)$} &
    \makecell{$(x+1)$\\$(x^9 + w^2x^8 + x^7$\\$+\,wx^2 + w^2x)$} &
    \makecell{$(x+1)$\\$(x^4 + x^3 + wx^2 + x + 1)$\\$(x^4 + x^3 + w^2x^2 + x + 1)$\\$(x^4 + w^2x^3 + x^2 + w^2x + 1)$} &
    $d = d_{\text{BKLC}} - 1$, see \cite{grassltable} \\
    \hline

    $[34,17,9]_4^H$ & 
    \makecell{$(x + 1)(x^4 + w^2x^3$\\$+\,x^2 + w^2x + 1)$} &
    \makecell{$x^{10} + x^8 + w^2x^7$\\$+\,w^2x^6 + x^2 + x$} &
    \makecell{$(x^4 + x^3 + wx^2 + x + 1)$\\$(x^4 + x^3 + w^2x^2 + x + 1)$\\$(x^4 + wx^3 + x^2 + wx + 1)$} &
    $d = d_{\text{BKHLCD}}$, see \cite{sok2020} \\
    \hline

    $[34,24,6]_4^H$ & $x + 1$ &
    \makecell{$(x + 1)(x + w^2)(x^2 + wx + w)$\\$(x^4 + x^3 + x + w^2)$} &
    \makecell{$(x + 1)$\\$(x^4 + wx^3 + x^2 + wx + 1)$\\$(x^4 + w^2x^3 + x^2 + w^2x+1)$} &
    $d = d_{\text{BKLC}}$, see \cite{grassltable} \\
    \hline

    $[34,26,5]_4^H$ & $1$ &
    \makecell{$x^7 + wx^5 + w^2x^4 + x^3$\\$+\,w^2x^2 + w$} &
    \makecell{$(x^4 + wx^3 + x^2 + wx + 1)$\\$(x^4 + w^2x^3 + x^2 + w^2x + 1)$} &
    $d = d_{\text{BKLC}}$, see \cite{grassltable} \\
    \hline

    $[38,29,5]_4^H$ & $1$ &
    \makecell{$w^2x^8 + w^2x^7 + w^2x^6$\\$+\,wx^5 + wx + 1$} &
    \makecell{$x^9 + wx^8 + wx^6 + wx^5$\\$+ w^2x^4 + w^2x^3 + w^2x + 1$} &
    $d = d_{\text{BKLC}} - 1$, see \cite{grassltable} \\
    \hline
\end{tabular}

\caption{Quaternary quasi-cyclic Hermitian LCD codes from Theorem \ref{qch}, length $14~\le~n~\le~38$.}
\label{tabhermitian1}
\end{table}

\begin{table}[H]
    \centering
\scriptsize
\begin{tabular}{|c|c|c|c|c|c|}
    \hline 
    $C$  & $g_{11}(x)$ & $g_{12}(x)$ & $g_{22}(x)$   & Parameters from \cite{crnkovic2023, grbac} \\
    \hline
    $[46,23,8]_4^H$ & $x + 1$ & 
    \makecell{$w^2x^{21} + x^{20} + wx^{15}$\\$+\,x^{14} + x^8 + w^2x^2$} &
    \makecell{$(x^{11} + x^9 + x^7 + x^6 + x^5 + x + 1)$\\$ (x^{11} + x^{10} + x^6 + x^5 + x^4 + x^2 + 1)$} &
    $[46,23,7]_4^H$\\
    \hline
    $[58,29,8]_4^H$ & $x + 1$ &
    \makecell{$x^{27} + w^2x^{19} + w^2x^{15}$\\$+\,x^{14} + x^{12} + wx^{10}$} &
    \makecell{$ x^{28} + x^{27} + x^{26} + x^{25} + x^{24} + x^{23} + x^{22} + x^{21}$\\$+\,x^{20} + x^{19} + x^{18} + x^{17} + x^{16} + x^{15} + x^{14} + x^{13}$\\$+\,x^{12} + x^{11} + x^{10} + x^9 + x^8 + x^7 + x^6 + x^5$\\$+\,x^4 + x^3 + x^2 + x + 1 $} &
    $[58,29,4]_4^H$\\
    \hline
    $[62,31,8]_4^H$ & $x + 1$ &
    \makecell{$x^{22} + x^{18} + x^{14} + wx^{12}$\\$+\,x + 1$} &
    \makecell{$(x^5 + x^2 + 1)(x^5 + x^3 + 1)(x^5 + x^3 + x^2 + x + 1)$\\$ (x^5 + x^4 + x^2 + x + 1)(x^5 + x^4 + x^3 + x + 1)$\\$ (x^5 + x^4 + x^3 + x^2 + 1)$} &
    $[62,31,7]_4^H$\\
    \hline
    $[70,35,8]_4^H$ & $x + 1$ &
    \makecell{$w^2x^{25} + wx^{22} + x^{19} + x^{18}$\\$+\,w^2x^{17} + x^{16}$} &
    \makecell{$(x^3 + x + 1)(x^3 + x^2 + 1)(x^4 + x^3 + x^2 + x + 1)$\\$ (x^{12} + x^{10} + x^9 + x^8 + x^7 + x^4 + x^2 + x + 1)$\\$ (x^{12} + x^{11} + x^{10} + x^8 + x^5 + x^4 + x^3 + x^2 + 1)$} &
    $[70,35,5]_4^H$\\
    \hline
    $[74,37,8]_4^H$ & $x + 1$ &
    \makecell{$x^{29} + w^2x^{27} + w^2x^{24}$\\$+\,x^8 + wx^5 + 1$} &
    \makecell{$ x^{36} + x^{35} + x^{34} + x^{33} + x^{32} + x^{31} + x^{30} + x^{29}$\\$+\,x^{28} + x^{27} + x^{26} + x^{25} + x^{24} + x^{23} + x^{22} + x^{21}$\\$+\,x^{20} + x^{19} + x^{18} + x^{17} + x^{16} + x^{15} + x^{14} + x^{13}$\\$+\,x^{12} + x^{11} + x^{10} + x^9 + x^8 + x^7 + x^6 + x^5$\\$+\,x^4 + x^3 + x^2 + x + 1$} &
    $[74,37,4]_4^H$\\
    \hline
\end{tabular}
\caption{Quaternary quasi-cyclic Hermitian LCD codes from Theorem \ref{qch}, length $46~\le~n~\le~74$.}
\label{tabhermitian2}
\end{table}


\section*{Conclusion}
In this work, we have provided polynomial-based characterizations of quasi-cyclic linear complementary dual (LCD) codes of index $2$ with respect to the Euclidean, Hermitian and symplectic inner products. Using these characterizations, we constructed several binary, ternary, and quaternary LCD codes. Our results extend the existing characterizations obtained for one-generator quasi-cyclic codes. Moreover, the techniques introduced in our characterization can be readily generalized to the broader class of quasi-twisted codes. A promising direction for future research is to obtain polynomial characterization for quasi-cyclic LCD codes of arbitrary index. This extension would require handling a large number of polynomials together, which will be more complicated.

\section*{Acknowledgment}

K. Abdukhalikov was supported by UAEU grants G00004233 and G00004614.
D. Ho was supported by  the Tromsø Research Foundation (project “Pure Mathematics in Norway”) and  UiT Aurora project MASCOT. 
G. K. Verma was supported by UAEU grant G00004614. 
S. Ling was supported by Nanyang Technological University Research Grant No. 04INS000047C230GRT01 and UAEU grant G00004233.

\end{document}